\documentclass[a4paper,11pt]{article}
\usepackage[utf8x]{inputenc}

\usepackage{oldgerm}
\usepackage{color}

\newcommand{\Ibb}[1]{ {\rm I\ifmmode\mkern -3.6mu\else\kern -.2em\fi#1}}
\newcommand{\ibb}[1]{\leavevmode\hbox{\kern.3em\vrule
     height 1.2ex depth -.3ex width .2pt\kern-.3em\rm#1}}
\newcommand{\Cl}{{\ibb C}}
\newcommand{\Rl}{{\Ibb R}}
\newcommand{\Nl}{{\Ibb N}}

\definecolor{lightgray}{rgb}{0.8,0.8,0.8}


\newcommand{\refitem}[1] {\textit{\ref{#1})}}

\newcommand{\Om}{\Omega}
\newcommand{\om}{\omega}
\newcommand{\omti}{\tilde{\omega}}

\newcommand{\te}{\theta}

\newcommand{\la}{\lambda}
\newcommand{\La}{\Lambda}
\newcommand{\eps}{\varepsilon}

\newcommand{\A}{\mathcal{A}}
\newcommand{\R}{\mathcal{R}}
\newcommand{\B}{\mathcal{B}}

\newcommand{\M}{\mathcal{M}}

\newcommand{\Hil}{\mathcal{H}}

\newcommand{\Q}{\mathcal{Q}}

\newcommand{\DD}{\mathcal{D}}
\newcommand{\W}{\mathcal{W}}

\newcommand{\Ss}{\mathscr{S}}   

\newcommand{\Le}{{\mathscr L}}

\newcommand{\BU}{\underline{\mathscr{S}}}

\newcommand{\pol}{\mathscr{P}}

\newcommand{\Nu}{\mathscr{N}}

\newcommand{\Cti}{\tilde{C}}

\newcommand{\Rti}{\tilde{R}}

\newcommand{\Mhat}{\widehat{\cal M}}

\newcommand{\fti}{\tilde{f}}

\newcommand{\hti}{\tilde{h}}
\newcommand{\lti}{\tilde{l}}
\newcommand{\uti}{\tilde{u}}

\newcommand{\rti}{\tilde{r}}
\newcommand{\Wti}{\tilde{W}}

\newcommand{\lhat}{\hat{l}}
\newcommand{\rhat}{\hat{r}}

\newcommand{\lcheck}{\check{l}}
\newcommand{\rcheck}{\check{r}}

\newcommand{\gti}{\tilde{g}}
\newcommand{\kti}{\widetilde{k}}

\newcommand{\LGpo}{\mathcal{L}_+^\uparrow}

\newcommand{\LG}{\mathcal{L}}

\newcommand{\PGpo}{\mathcal{P}_+^\uparrow}   

\newcommand{\PG}{\mathcal{P}}

\def\bp{{\boldsymbol{p}}}
\def\bx{{\boldsymbol{x}}}
\def\by{{\boldsymbol{y}}}
\def\bz{{\boldsymbol{z}}}

\def\bk{{\mbox{\boldmath{$k$}}}}
\def\bl{{\boldsymbol{l}}}
\def\br{{\boldsymbol{r}}}

\def\bmu{{\boldsymbol{\mu}}}

\newcommand{\OO}{O}

\newcommand{\supp}{\text{supp}\,}

\newcommand{\dom}{\mathrm{dom}\,}

\newcommand{\im}{\mathrm{Im}}

\newcommand{\ot}{\otimes}

\newcommand{\oth}{\,\hat{\otimes}\,}

\newcommand{\zd}{z^{\dagger}}

\newcommand{\ad}{a^{\dagger}}

\usepackage{amsmath}
\usepackage{amsfonts}
\usepackage{amssymb}
\usepackage{mathrsfs}
\usepackage[pdfborder={0 0 0}]{hyperref}
\usepackage{color}
\usepackage[T1]{fontenc}
\usepackage{paralist}

\newtheorem{theorem}{Theorem}[section]
\newtheorem{proposition}[theorem]{Proposition}
\newtheorem{lemma}[theorem]{Lemma}

\newtheorem{definition}[theorem]{Definition}

\newenvironment{proof}{\medskip \noindent {\em Proof:}}{\hfill $\square$ \\[2mm] \indent}

\newenvironment{propositionlist}{\begin{compactenum}[\itshape i.)]}{\end{compactenum}}

\numberwithin{equation}{section}

\newlength{\dinwidth}
\newlength{\dinmargin}

\title{Deformations of quantum field theories and integrable models}
\author{Gandalf Lechner\footnote{Supported by FWF project P22929--N16 "Deformations of Quantum Field Theories".}
\\
\vspace*{-1mm}
\\
\footnotesize Faculty of Physics, University of Vienna, \\
\footnotesize Boltzmanngasse 5, A-1090 Vienna, Austria\\
\footnotesize \tt gandalf.lechner@univie.ac.at
}

\begin{document}

\maketitle

\begin{abstract}
	Deformations of quantum field theories which preserve Poincar\'e covariance and localization in wedges are a novel tool in the analysis and construction of model theories. Here a general scenario for such deformations is discussed, and an infinite class of explicit examples is constructed on the Borchers-Uhlmann algebra underlying Wightman quantum field theory. These deformations exist independently of the space-time dimension, and contain the recently studied warped convolution deformation as a special case. In the special case of two-dimensional Minkowski space, they can be used to deform free field theories to integrable models with non-trivial S-matrix.
\end{abstract}

\section{Introduction}\label{section:introduction}

In the last years, many new quantum field theoretic models have been constructed with non-standard methods \cite{Schroer:1997-1,SchroerWiesbrock:2000-1,BrunettiGuidoLongo:2002,Lechner:2003,LongoRehren:2004,Lechner:2005,MundSchroerYngvason:2006,GrosseLechner:2007,BuchholzSummers:2007,BuchholzSummers:2008,GrosseLechner:2008,BuchholzLechnerSummers:2010,LongoWitten:2010,DybalskiTanimoto:2010}. Among the different approaches used for constructing these models, a recurring theme is to start with a well-understood model (like a free field theory), and then apply some kind of deformation to change it to a model with non-trivial interaction. As is well known, it is extremely complicated to carry out such a procedure on a non-perturbative level when requiring that it should keep the full covariance, spectral and locality properties of quantum field theory intact. However, interesting manageable examples do exist when the locality requirements are somewhat weakened.

More precisely, there exist many models of quantum fields which are not point-like localized, but rather localized in certain unbounded, wedge-shaped regions (wedges) in Minkowski space \cite{SchroerWiesbrock:2000-1,BrunettiGuidoLongo:2002,Lechner:2003,LongoRehren:2004,BuchholzSummers:2007,GrosseLechner:2007,BuchholzLechnerSummers:2010}. These models are still fully Poincar\'e covariant and comply with Einstein causality inasmuch that observables with spacelike separated wedges commute. Using the algebraic framework of quantum field theory \cite{Haag:1996}, it is also in principle possible \cite{BuchholzLechner:2004} to extract all observables localized in bounded spacetime regions. Moreover, the localization in wedges is sharp enough to consistently compute the two-particle scattering matrix \cite{BorchersBuchholzSchroer:2001}, and decide if the constructed model exhibits non-trivial interaction.

In view of these facts, wedge-local quantum field theories have many of the characteristic features of fully local quantum field theories, and understanding their structure is an important intermediate step in the rigorous construction of interacting models. It is therefore interesting to note that it is possible to construct wedge-local quantum field theories non-perturbatively, and introduce non-trivial interaction by deformation techniques.

A particular deformation of this kind, based on actions of the translation group, is by now well understood. After its first appearance in the context of deformed free field theories on non-commutative Minkowski space \cite{GrosseLechner:2007}, it was generalized to an operator-algebraic setting in \cite{BuchholzSummers:2008}, where it is known as warped convolution. In the framework of Wightman field theories, this deformation manifests itself as a deformation of the tensor product of the testfunction algebra \cite{GrosseLechner:2008}, and later on, the connection to Rieffel's strict deformation quantization \cite{Rieffel:1992} was explored \cite{BuchholzLechnerSummers:2010}. By now, the warped convolution technique has also successfully been applied to the deformation of conformal field theories \cite{DybalskiTanimoto:2010} and quantum field theories on curved spacetimes \cite{DappiaggiLechnerMorfaMorales:2010}.

In this paper we start to explore more general deformations of wedge-local quantum field theories. As a first scenario for such deformations, we focus here on Wightman quantum field theories \cite{StreaterWightman:1964,Jost:1965}. Any Wightman quantum field theory is given by a specific representation of the tensor algebra $\BU$ over Schwartz' function space $\Ss(\Rl^d)$. The deformations studied here are based on linear homeomorphisms $\rho:\BU\to\BU$ commuting with the natural Poincar\'e automorphisms $\alpha_{x,\La}$ on $\BU$, for $(x,\La)$ in a subgroup of the Poincar\'e group which models the geometry of a reference wedge. We then equip $\BU$ with a family of new products, namely $f,g\mapsto \rho^{-1}(\rho(f)\ot\rho(g))$, and Lorentz transforms thereof. Every single of these products provides only a trivial deformation of the tensor product $\ot$, but their interplay with the local structure of $\BU$ gives rise to non-trivial deformations of a net of algebras localized in wedges. If a compatibility condition between $\rho$ and a state $\om$ on $\BU$ is met, one can pass to suitable GNS representations, where all twisted product structures are represented on the same Hilbert space. Here we obtain new quantum field theoretic models, which are wedge-local under further conditions on the deformation map $\rho$ and the state $\om$.

In Section \ref{section:generaldeformations}, we explain these deformations in a general setting. The main task of finding interesting examples of deformation maps is taken up in Section \ref{section:examples}. Here we consider a simple class of such mappings $\rho$, given by sequences of $n$-point functions, and their compatible states. We show that by carefully adjusting these $n$-point functions, one arrives at an infinite class of deformations, leading to new Poincar\'e covariant and wedge-local model theories in any number of space-time dimensions. These models are investigated in more detail in Section \ref{section:FockSpace}, where Hilbert space representations of deformed quantum fields are presented, and it is shown that they describe non-trivial interaction. The two-particle S-matrix can be calculated explicitly, and depends on the deformation parameter.

The representing quantum fields are typically unbounded operators. In Section \ref{section:modular} we show how to pass from these fields to associated von Neumann algebras, and analyze their Tomita-Takesaki modular structure.

In Section \ref{section:IntegrableModels}, we consider the special case of two-dimensional Minkowski space. Here our construction yields a known family of completely integrable quantum field theories. It is shown that the structure of the deformation maps implies characteristic features of the S-matrix, such as its analyticity and crossing properties. Section \ref{section:Conclusions} contains our conclusions.

\section{Deformation maps on the Borchers-Uhlmann algebra}\label{section:generaldeformations}

In this section, we formulate a general deformation scenario for Wightman quantum field theories, based on the tensor algebra $\BU$ over the Schwartz space $\Ss(\Rl^d)$. We will assume that the space-time dimension $d\geq1+1$ is even, as this slightly simplifies our discussion in some places. Most of the following can also be formulated in a vastly more general setting of quite general topological ${}^*$-algebras, but since the examples to be discussed later make use of the specific structure of $\BU$, we restrict our considerations to this particular algebra also in this section.

Let us first recall the structure of the Borchers-Uhlmann algebra $\BU$ \cite{Borchers:1962, Uhlmann:1962}: As a topological vector space, $\BU=\bigoplus_{n=0}^\infty \Ss_n$ is the locally convex direct sum of the Schwartz spaces $\Ss_n:=\Ss(\Rl^{n d})$, $n\geq0$, with $\Ss_0:=\Cl$. Elements of $\BU$ are thus terminating sequences $f=(f_0,f_1,f_2,...,f_N,0,...\,)$, $f_n\in\Ss_n$. Equipped with the tensor product
\begin{align}\label{eq:ProductBU}
	(f\ot g)_n(x_1,...,x_n)
	:=
	\sum_{k=0}^n f_k(x_1,...,x_k)\cdot g_{n-k}(x_{k+1},...,x_n)
	\,,\qquad
	x_1,...,x_n\in\Rl^d\,,
\end{align}
${}^*$-involution
\begin{align}\label{eq:StarBU}
	{f^*}_n(x_1,...,x_n)
	:=
	\overline{f_n(x_n,...,x_1)}
	\,,
\end{align}
and unit $1_n:=\delta_{n,0}$, the linear space $\BU$ becomes a unital topological ${}^*$-algebra.

On $\BU$, the proper orthochronous Poincar\'e group acts by the continuous automorphisms
\begin{align}\label{eq:AlphaALa}
	(\alpha_{a,\La}f)_n(x_1,...,x_n)
	:=
	f_n(\La^{-1}(x_1-a),...,\La^{-1}(x_n-a))
	\,,\qquad
	(a,\La)\in\PGpo\,.
\end{align}
For our purposes, it is advantageous to implement time-reversing Lorentz
transformations by {\em anti}\,linear maps on $\BU$. In particular, the reflection
$j(x^0,...,x^{d-1}):=(-x^0,-x^1,x^2,...,x^{d-1})$ acts on $\BU$ according to
\begin{align*}
	(\alpha_j f)_n(x_1,...,x_n)
	:=
	\overline{f_n(jx_1,...,jx_n)}
	\,,
\end{align*}
and yields an extension of $\alpha$ to an automorphic action of the proper Poincar\'e group $\PG_+$ on $\BU$ (antilinear for $\PG_+^\downarrow$).

We define the support $\supp f$ of an element $f\in\BU$ as the smallest closed set $\OO$ in $\Rl^d$ such that $\supp f_n\subset \OO^{\times n}$ for all $n\geq1$. Given $\OO\subset\Rl^d$ and $f_n\in\Ss_n$, we will also write $\supp f_n\subset \OO$ or $f_n\in\Ss_n(\OO)$ instead of $\supp f_n\subset\OO^{\times n}$. With this definition of support, the set $\BU(\OO):=\{f\in\BU\,:\,\supp f\subset \OO\}$ is a unital ${}^*$-subalgebra of $\BU$, for any $\OO\subset\Rl^d$. Since $\supp \alpha_{x,\La}(f)=\La\,\supp f +x$, the automorphisms $\alpha_{x,\La}$ act covariantly on the net $\OO\mapsto\BU(\OO)$,
\begin{align}
	\alpha_{x,\La}(\BU(\OO))
	=
	\BU(\La\OO+x)
	\,.
\end{align}
This net becomes {\em local}, i.e., subalgebras $\BU(\OO_1)$, $\BU(\OO_2)\subset\BU$ associated with spacelike separated regions $\OO_1\subset\OO_2'$ commute, after dividing by the so-called {\em locality ideal} \cite{Borchers:1962, Yngvason:1984}, the two-sided ideal $\Le\subset\BU$ generated by all commutators $f_1\ot g_1-g_1\ot f_1$, with $f_1,g_1\in\Ss_1$ having spacelike separated supports.

We will consider states on $\BU$ subsequently, and introduce here some notation regarding GNS representations. For a state $\om$ on $\BU$, we write $(\Hil_\om,\phi_\om,\Om_\om)$ for the GNS triple associated with $(\BU,\om)$, and $\DD_\om:=\phi_\om(\BU)\Om_\om\subset\Hil_\om$ for the (dense) domain of the representing field operators. The equivalence classes $\{f+g\in\BU\,:\,\om(g^*\ot g)=0\}$ will be denoted $\Psi_\om(f)\in\DD_\om$. Thus $\Psi_\om(1)=\Om_\om$, and the fields act on $\DD_\om$ according to $\phi_\om(f)\Psi_\om(g)=\Psi_\om(f\ot g)$. As $\phi_\om$ is a representation, we have $\phi_\om(f)\phi_\om(g)=\phi_\om(f\ot g)$ and $\phi_\om(f)^*\supset\phi_\om(f^*)$. The represented localized field algebras are denoted $\pol_\om(\OO):=\phi_\om(\BU(\OO))$, $\OO\subset\Rl^d$.
\\
\\
Quantum field theories arise from the tensor algebra $\BU$ as GNS-representations in suitable states \cite{StreaterWightman:1964}. For a state $\om$ which vanishes on the locality ideal $\Le$, field operators $\phi_\om(f)$ and $\phi_\om(g)$ commute on $\DD_\om$ if the supports of $f$ and $g$ are spacelike separated. If $\om$ is also invariant under the automorphisms $\alpha_{x,\La}$ (invariant up to a conjugation for time-reversing $\La$), there also exists an (anti-)unitary representation $U_\om$ of the proper Poincar\'e group on $\Hil_\om$ which implements the automorphisms $\alpha_{x,\La}$. In this case, we obtain the familiar structure of a covariant net of local ${}^*$-algebras:
\begin{align}\label{eq:Net}
	\nonumber
	\pol_\om(\OO_1) \subset \pol_\om(\OO_2)\;&\text{ for }\;\OO_1\subset\OO_2\,,\\
	U_\om(x,\La)\pol_\om(\OO)U_\om(x,\La)^{-1} &= \pol_\om(\La \OO+x)\,,\\
	[\pol_\om(\OO_1),\,\pol_\om(\OO_2)]\DD_\om &=0\;\text{ for }\; \OO_1\subset\OO_2'
	\nonumber
	\,,
\end{align}
where $\OO_2'$ denotes the causal complement of $\OO_2$ in $\Rl^d$. For vacuum states, one is interested in the situation where the translations $x\mapsto U_\om(x,1)$ fulfill the spectrum condition. In this case, also a Reeh-Schlieder property holds, i.e., the subspace $\pol_\om(\OO)\Om_\om$ is dense in $\Hil_\om$ for any open region $\OO\subset\Rl^d$.

On a technical level, note that the field operators $\phi_\om(f)$, $f\in\BU$, are densely defined on the common $U_\om$-invariant domain $\DD_\om$ and closable, but in general unbounded. Several conditions on $\om$ are known which imply that one can pass from such a net of unbounded operators to nets of von Neumann algebras on $\Hil_\om$ \cite{BorchersZimmermann:1963,DriesslerSummersWichmann:1986,Buchholz:1990-1,BorchersYngvason:1990}. We will however not deal with this question here, and consider only algebras of unbounded operators.
\\
\\
The construction of states which annihilate $\Le$ and satisfy the spectrum condition has proven to be extremely difficult. In more than two space-time dimensions, only states leading to (generalized) free field theories are known. In view of these difficulties, we will not attempt a direct construction of quantum field theories by finding suitable states $\om$ on $\BU$, but rather use a deformation approach.

To explain this approach, we first recall that in the construction of many models discussed in the recent literature \cite{Borchers:1992,BrunettiGuidoLongo:2002,Lechner:2003,LongoRehren:2004,BuchholzSummers:2007,GrosseLechner:2007,Lechner:2008,BuchholzLechnerSummers:2010,DybalskiTanimoto:2010,Mund:2010}, a specific weakened version of the net structure \eqref{eq:Net} plays a prominent role. Instead of algebras $\pol_\om(\OO)$ associated with arbitrarily small spacetime regions $\OO$, one considers only specific regions, so-called {\em wedges}. Recall that the {\em right wedge} is the region $W_0:=\{(x^0,...,x^{d-1})\in\Rl^d\,:\,x^1>|x^0|\}$, and the set $\W$ of all wedges is the Poincar\'e orbit of $W_0$, i.e.,  $\W:=\{\La W_0+x\,:\,(x,\La)\in\PG_+\}$. In particular, the causal complement $W'$ of a wedge $W\in\W$ is also contained in $\W$, and for our reference region $W_0$, there holds $W_0'=-W_0=jW_0$, with $j(x^0,...,x^{d-1}):=(-x^0,-x^1,x^2,...,x^{d-1})$ the reflection at the edge of $W_0$.

In the context of the GNS data $\DD_\om\subset \Hil_\om$, $U_\om$ described before, a {\em wedge-local quantum field theory} is defined to be a collection of ${}^*$-algebras $\pol_\om(W)$, $W\in\W$, consisting of operators defined on $\DD_\om$, such that the properties \eqref{eq:Net} hold for $\OO,\OO_1,\OO_2\in\W$. Since $\W$ consists only of a single Poincar\'e orbit, such a net can be equivalently characterized in terms of a {\em single} algebra \cite{BuchholzSummers:2008,BaumgrtelWollenberg:1992} $\pol_0$ of operators acting on $\Hil_\om$ by requiring
\begin{align}
	U_\om(x,\La)\pol_0 U_\om(x,\La)^{-1}&\subset \pol_0\;\text{ for }\;\La W_0+x\subset W_0
	\label{eq:WedgeAlgebraCovariance}
	\,,\\
	[U_\om(0,j)\pol_0 U_\om(0,j),\,\pol_0]\DD_\om &=0\,.
	\label{eq:WedgeAlgebraLocality}
\end{align}
It is then straightforward to verify that $\pol(\La W_0+x):=U(x,\La)\pol_0 U(x,\La)^{-1}$ defines a wedge-local quantum field theory (A {\em simple causal net} in the terminology of \cite{BaumgrtelWollenberg:1992}).

Clearly any net $\OO\mapsto\pol_\om(\OO)$ \eqref{eq:Net} also defines such a wedge algebra $\pol_0$. But, as we shall see, an algebra $\pol_0$ satisfying the conditions \eqref{eq:WedgeAlgebraCovariance} and \eqref{eq:WedgeAlgebraLocality} with respect to a given representation $U_\om$ of $\PG_+$ is much easier to construct than a full net \eqref{eq:Net}. Moreover, after passing to a net of von Neumann algebras, one can in principle extract algebras of observables localized in arbitrary spacetime regions from these data \cite{BuchholzSummers:2008,Borchers:1992,BuchholzLechner:2004}.

In the deformation approach, one takes the point of view that a fully local and covariant quantum field theory in the sense of \eqref{eq:Net} is given. These data will usually be realized by free field theories, and in particular define an operator algebra $\pol_0$ and a representation $U_\om$ in a suitable relative position on some Hilbert space $\Hil_\om$. One then keeps $\Hil_\om$ and $U_\om$ fixed, and changes (deforms) the algebra $\pol_0$ in such a manner that \eqref{eq:WedgeAlgebraCovariance} and \eqref{eq:WedgeAlgebraLocality} remain valid. For suitably chosen deformations, this process leads to inequivalent nets, and in particular turns interaction-free theories into models with non-trivial interaction.
\\
\\
To find examples of deformations preserving the two conditions \eqref{eq:WedgeAlgebraCovariance} and \eqref{eq:WedgeAlgebraLocality}, one possible approach is to take the point of view that a deformation of an algebra is a deformation of the product of that algebra. This is the approach taken in the deformation theory of algebras in the mathematics literature \cite{Gerstenhaber:1964}, which has already led to deformations of quantum field theories in certain examples \cite{GrosseLechner:2008,BuchholzSummers:2008,BuchholzLechnerSummers:2010}.

By a product on $\BU$, we will always mean a bilinear separately continuous map $f,g\mapsto f\oth g$, which is associative and moreover compatible with the unit and star involution in $\BU$, i.e.,
\begin{align}
	f\oth 1&=f= 1\oth f\,,\qquad f\in\BU,\\
	(f\oth g)^* &= g^*\oth f^*\,,\qquad \;\;f,g\in\BU\,.
\end{align}
The structure of the family of such products clearly depends on the structure of the algebra under consideration. In the situation at hand, where $\BU$ is a tensor algebra, it is known that $\BU$ is rigid in the sense of algebraic deformation theory \cite{Gerstenhaber:1964}. That is, all products $\oth$ on $\BU$ are of the form
\begin{align}\label{eq:RhoProduct}
	f\oth g
	=
	\rho^{-1}(\rho(f)\ot \rho(g))
	=:
	f\ot_\rho g
	\,,
\end{align}
where $\rho:\BU\to\BU$ is a linear homeomorphism with $\rho(1)=1$ and $\rho(f)^*=\rho(f^*)$. Clearly, the algebra $\BU^\rho:=(\BU,\ot_\rho)$ given by the linear space $\BU$, endowed with the product $\ot_\rho$, and unchanged unit and involution, is isomorphic as a unital ${}^*$-algebra to $\BU$. This is the reason why products of the form \eqref{eq:RhoProduct} are considered trivial in the deformation theory of single algebras \cite{Gerstenhaber:1964}, and $\BU$ is rigid. But we will see later that the use of such trivial deformations will result in non-trivial deformations of nets of wedge algebras nonetheless, as also the local structure of $\BU$ matters here.

For deformations compatible with localization in wedges, the invariance property \eqref{eq:WedgeAlgebraCovariance} suggests to require a certain amount of compatibility between the deformation map $\rho$ and the Poincar\'e action $\alpha$. We therefore make the following definition.

\begin{definition}\label{definition:DeformationMap}
A deformation map (relative to $W_0$) is a linear homeomorphism $\rho:\BU\to\BU$ such that
\begin{enumerate}
	\item $\rho(1)=1$.
	\item $\rho(f)^*=\rho(f^*)$, $f\in\BU$.
	\item $\rho\circ\alpha_{x,\La}=\alpha_{x,\La}\circ \rho$ for all $(x,\La)\in\PG_+$ with $\La W_0+x\subset W_0$.
\end{enumerate}
\end{definition}

We remark that the third condition in this definition is equivalent to
\begin{align}\label{eq:RhoCovarianceExplicit}
	\rho\circ\alpha_{x,\La}=\alpha_{x,\La}\circ \rho
	\qquad \text{ for all } x\in\Rl^d\text{ and all }\La \text{ with } \La W_0=W_0\,.
\end{align}
This is due to the special form of the wedge regions: First, there holds $W_0+x\subset W_0$ for all $x\in\overline{W_0}$. Hence $\alpha_x\circ\rho=\rho\circ\alpha_x$ for all $x\in\overline{W_0}$. Multiplying by $\alpha_{-x}$ from both sides, we see that this equation also holds for $x\in-\overline{W_0}$. As any $y\in\Rl^d$ can be written as $y=x+x'$ with $x\in\overline{W_0}$ and $x'\in-\overline{W_0}$, this implies that $\rho$ must commute with {\em all} translations. Second, if a Poincar\'e transformation $(x,\La)$ maps $W_0$ inside itself, then necessarily $\La W_0=W_0$ \cite{ThomasWichmann:1997}. This explains the equivalence of \eqref{eq:RhoCovarianceExplicit} with Definition \ref{definition:DeformationMap} {\em iii)}. To summarize, a deformation map has to preserve the linear, topological, unital, and ${}^*$-structure of $\BU$, and commute with the automorphisms $\alpha_{x,\La}$ for $(x,\La)$ in a specific subgroup of $\PG_+$, which models the geometry of the wedge region $W_0$.

The properties required in Definition \ref{definition:DeformationMap} are stable under composition and taking inverses. With identity as the identity map on $\BU$, the deformation maps therefore form a group $\R$. In deformation theory, one is usually interested in studying certain one-parameter families $\rho_\la\in\R$, $\la\in\Rl$, such that $\la\mapsto\rho_\la$ is continuous in an appropriate sense, and $\rho_0=\rm id$. We will see examples of such one parameter families in Section \ref{section:examples}. For the present general considerations, it will be sufficient to consider deformation maps as such, without introducing a deformation parameter.
\\
\\
Given any deformation map, the product \eqref{eq:RhoProduct} will be referred to as the {\em associated deformed product} on $\BU$. In view of Definition \ref{definition:DeformationMap} {\em iii)}, the maps $\alpha_\La$ with $\La W_0=W_0$  act as automorphisms also with respect to the product $\ot_\rho$. For general $\La\in\LG_+$, one has $\alpha_\La(f\ot_\rho g)=\alpha_\La(f)\ot_{\rho_\La}\alpha_\La(g)$, where $\rho_\La:=\alpha_\La \circ \rho \circ  \alpha_\La^{-1}$ is a deformation map relative to $\La W_0$, and in general $\rho_\La\neq\rho$. We therefore obtain a whole family of products $\ot_{\rho_\La}$, parametrized by the Lorentz group modulo the stabilizer group of the wedge. This family includes in particular the {\em opposite deformation map}
\begin{align}\label{eq:OppositeRho}
	\rho'
	:=
	\alpha_j \circ \rho\circ\alpha_j
	\,.
\end{align}
To construct a wedge-local quantum field theory from this family of deformed products, we have to represent all the algebras $(\BU,\ot_{\rho_\La})$, $\La\in\LG_+$, on a common Hilbert space. This is possible in specific GNS representations.

\begin{definition}\label{definition:CompatibleState}
	A state $\om$ on $\BU$ is called compatible with a deformation map $\rho$ if
	\begin{align}\label{eq:RhoCompatabilityOmega}
		\om(f\ot_\rho g)
		&=
		\om(f\ot g)
		\,,\qquad
		f,g\in\BU\,.
	\end{align}
\end{definition}

Note that this definition does {\em not} imply that multiple deformed products reduce to undeformed products in $\om$, i.e., in general $\om(f_1 \ot_\rho ... \ot_\rho f_n)\neq\om(f_1\ot ... \ot f_n)$ for $n>2$. We are interested in compatible states because they produce common representation spaces for deformed and undeformed tensor products via the GNS construction.

\begin{proposition}\label{proposition:GNSRepresentationOfCompatibleDeformation}
	Let $\rho$ be a deformation map and $\om$ a $\rho$-compatible state. Then $\om$ is also a state on $\BU^\rho$, and the GNS triples $(\Hil_\om,\phi_\om,\Om_\om)$ of $(\BU,\om)$ and $(\Hil^\rho_\om,\phi^\rho_\om,\Om^\rho_\om)$ of $(\BU^\rho,\om)$ are related by
	\begin{align}
		\Hil^\rho_\om
		&=
		\Hil_\om\,,\\
		\Om^\rho_\om
		&=
		\Om_\om\,,\\
		\phi^\rho_\om(f)\phi_\om(g)\Om_\om
		&=
		\phi_\om(f\ot_\rho g)\Om_\om
		\,,\qquad
		f,g\in\BU\,.\label{eq:PhiRhoOnH}
	\end{align}
\end{proposition}
\begin{proof}
	The state $\om$ on $\BU$ clearly defines a normalized linear functional $f\mapsto\om(f)$ on $\BU^\rho$. The $\rho$-compatibility and positivity of $\om$ imply, $f\in\BU$,
	\begin{align*}
		\om\left(f^*\ot_\rho f\right)
		&=
		\om(f^*\ot f)
		\geq
		0\,.
	\end{align*}
	Hence $\om$ is also a state on the deformed algebra $\BU^\rho$.

	To verify the statements about the GNS representations of $(\BU,\om)$ and $(\BU^\rho,\om)$, let $\Nu_\om:=\{f\in\BU\,:\,\om(f^*\ot f)=0\}$ and $\Nu_\om^\rho:=\{f\in\BU\,:\,\om(f^*\otimes_\rho f)=0\}$ denote the respective Gelfand ideals. Since $\om$ is $\rho$-compatible, we have $\om(f^*\otimes_\rho f)=\om(f^*\ot f)$, and hence $\Nu^\rho_\om=\Nu_\om$ as linear spaces. As also $\BU$ and $\BU^\rho$ coincide as linear spaces, we have $\BU/\Nu_\om=\BU^\rho/\Nu_\om^\rho$. By the $\rho$-compatability of $\om$, these pre-Hilbert spaces carry the same scalar product $\langle\Psi_\om(f),\Psi_\om(g)\rangle=\om(f^*\ot g)=\om(f^*\otimes_\rho g)$, which implies in particular that their Hilbert space closures $\Hil_\om$ and $\Hil^\rho_\om$ are identical. The implementing vectors $\Om_\om$ and $\Om_\om^\rho$ are both equal to the equivalence class $\Psi_\om(1)=\Psi_\om^\rho(1)$ and therefore identical.

	The GNS representation $\phi^\rho_\om$ of $\BU^\rho$ acts on this space according to, $f,g\in\BU$,
	\begin{align}
		\phi^\rho_\om(f)\phi_\om(g)\Om_\om
		=
		\phi^\rho_\om(f)\Psi_\om(g)
		=
		\Psi_\om(f\ot_\rho g)
		=
		\phi_\om(f\otimes_\rho g)\Om_\om
		\,.
	\end{align}
	This is well-defined since $\Nu_\om$ is, by the preceding argument, also a left ideal with respect to the deformed product, and the proof is finished.
\end{proof}

We now explain how wedge-local quantum field theories can be constructed from deformation maps $\rho$. To this end, suppose $\rho$ is a deformation map, and $\om$ is a $\rho$-compatible state which is invariant under $\alpha$ in the sense that, $f\in\BU$,
\begin{align}
	\om(\alpha_{x,\La}(f))
	&=
	\left\{
	\begin{array}{rcl}
		\overline{\om(f)} &;& (x,\La)\in\PG_+^\downarrow\\
		\om(f) &;& (x,\La)\in\PGpo
	\end{array}
	\right.
	\,.
\end{align}
We then have, $f,g\in\BU$, $\La\in\LG_+^\uparrow$,
\begin{align*}
	\om(f\otimes_{\rho_\La}g)
	&=
	\om(\alpha_\La(\alpha_\La^{-1}(f)\otimes_\rho \alpha_\La^{-1}(g)))
	=
	\om(\alpha_\La^{-1}(f)\otimes_\rho \alpha_\La^{-1}(g))
	\\
	&=
	\om(\alpha_\La^{-1}(f)\ot\alpha_\La^{-1}(g))
	=
	\om(f\ot g)
	\,,
\end{align*}
and by an analogous calculation, also $\om(f\otimes_{\rho_\La}g)=\om(f\ot g)$ for $\La\in\LG_+^\downarrow$. Hence the state $\om$ is compatible with all Lorentz transformed deformation maps $\rho_\La$, $\La\in\LG_+$. In view of Proposition \ref{proposition:GNSRepresentationOfCompatibleDeformation}, all these deformations are thus realized on the GNS space of the undeformed algebra, and can be compared in terms of the Hilbert space operators $\phi^{\rho_\La}_\om(f)$ on $\DD_\om\subset \Hil_\om$.

It is clear from our construction that the (anti-)unitary representation $U_\om$ implementing $\alpha$ on $\Hil_\om$ satisfies $U_\om(x,\La)\Psi_\om(f)=\Psi_\om(\alpha_{x,\La}\,f)$, $(x,\La)\in\PG_+$, $f\in\BU$. After a small calculation making use of $\phi^{\rho_\La}_\om(f)\Psi_\om(g)=\Psi_\om(f\ot_{\rho_\La}g)$ \eqref{eq:PhiRhoOnH}, this yields the transformation law
\begin{align}
	U_\om(x,\La)\phi^\rho_\om(f)U_\om(x,\La)^{-1}
	=
	\phi_\om^{\rho_\La}(\alpha_{x,\La}f)
	\,,\qquad
	(x,\La)\in\PG_+,\;f\in\BU\,.
\end{align}
In particular, for those transformations $(x,\La)$ that satisfy $\La W_0+x\subset W_0$, the corresponding automorphisms commute with $\rho$ (Definition \ref{definition:DeformationMap} {\em iii)}), and we have
\begin{align}\label{eq:UOmPhiOmCovariance}
	U_\om(x,\La)\phi^\rho_\om(f)U_\om(x,\La)^{-1}
	=
	\phi_\om^\rho(\alpha_{x,\La}f)
	\,,\qquad
	\La W_0+x\subset W_0\,,\;f\in\BU\,.
\end{align}

To produce a wedge-localized algebra complying with \eqref{eq:WedgeAlgebraCovariance} and \eqref{eq:WedgeAlgebraLocality}, we have to use elements $f\in\BU$ with support in $W_0$. As the deformation map $\rho$ will usually not preserve supports, $\BU(W_0)$ will not be an algebra with respect to the deformed product $\ot_\rho$. We therefore consider the ${}^*$-algebra $\pol^\rho_{\om,0}$ {\em generated by} all $\phi^\rho_\om(f)$, $f\in\BU(W_0)$. The transformation law \eqref{eq:UOmPhiOmCovariance} then implies the desired invariance  \eqref{eq:WedgeAlgebraCovariance} of $\pol_{\om,0}^\rho$.

The crucial locality condition \eqref{eq:WedgeAlgebraLocality} is equivalent to the vanishing of the commutators
\begin{align}\label{eq:PhiRhoPhiRho'-Commutator}
	[\phi_\om^{\rho}(f),\,\phi_\om^{\rho'}(g')]\Psi=0
	\,,\qquad
	f\in\BU(W_0),\;g\in\BU(W_0')\,,\Psi\in\DD_\om\,.
\end{align}
We will say that a deformation map $\rho$ {\em is wedge-local in a state $\om$} which is compatible with $\rho$ and $\rho'$ if \eqref{eq:PhiRhoPhiRho'-Commutator} holds. In this case, the algebra $\pol_{\om,0}^\rho$ complies with \eqref{eq:WedgeAlgebraCovariance} and \eqref{eq:WedgeAlgebraLocality}, and can therefore be used to generate a quantum field theory model.
\\
\\
To illustrate the conditions on the interplay of $\rho$ and $\om$, we recall that the deformation map given by warped convolution \cite{GrosseLechner:2008} is compatible with all translationally invariant states on $\BU$. But the locality condition \eqref{eq:PhiRhoPhiRho'-Commutator} is only valid if $\om$ annihilates $\Le$, the translations $U_\om(x,1)$ satisfy a spectrum condition, and the parameters defining $\rho$ are suitably chosen \cite{GrosseLechner:2008,BuchholzLechnerSummers:2010}. Hence the validity of \eqref{eq:PhiRhoPhiRho'-Commutator} is not a property of $\rho$ alone, but also involves properties of $\om$ going beyond compatibility and vanishing on $\Le$.

In the present generality, it seems to be difficult to find manageable conditions on $\rho$ and $\om$ which imply that $\rho$ is wedge-local in $\om$. We will therefore present in the next section a family of explicit deformation maps $\rho$ together with their compatible states $\om$ such that $\rho$ is wedge-local in $\om$. Before moving on to the examples, we point out that the wedge-locality condition amounts to the vanishing of matrix elements of commutators with respect to the undeformed product, a result that will be useful later on.

\begin{lemma}\label{lemma:LocalityInOmega}
	Let $\rho$ be a deformation map and $\om$ a state on $\BU$ which is compatible with $\rho$ and $\rho'$. Then $\rho$ is wedge-local in $\om$ if and only if
	\begin{align}\label{eq:WedgeLocalityInOmega}
		\om((h\ot_\rho f)\ot(g'\ot_{\rho'} k))
		&=
		\om((h\ot_{\rho'} g')\ot(f\ot_{\rho} k))
	\end{align}
	for all $f\in \BU(W_0)$, $g'\in\BU(W_0')$, $h,k\in\BU$.
\end{lemma}
\begin{proof}
	For $\rho$ to be wedge-local in $\om$, we need to show $[\phi_\om^\rho(f),\,\phi^{\rho'}_\om(g')]\Psi=0$ for all $\Psi\in\DD_\om$, $f\in \BU(W_0)$, $g'\in\BU(W_0')$. Since $\DD_\om=\phi_\om(\BU)\Om_\om$ is dense in $\Hil_\om$, this is equivalent to the vanishing of the matrix elements $\langle\Om_\om,\,\phi_\om(h)[\phi_\om^\rho(f),\,\phi^{\rho'}_\om(g')]\phi_\om(k)\Om_\om\rangle=0$ for arbitrary $h,k\in\BU$. But in view of the compatibility of $\rho,\rho'$ with $\om$, and the associativity of the products $\ot_\rho$, $\ot_{\rho'}$, we can rewrite these matrix elements as
	\begin{align*}
		0
		&=
		\langle\Om_\om,\,\phi_\om(h)[\phi_\om^\rho(f),\,\phi^{\rho'}_\om(g')]\phi_\om(k)\Om_\om\rangle
		\\
		&=
		\langle\Om_\om,\,\phi_\om(h)\phi_\om(f\ot_\rho(g'\ot_{\rho'}k))\Om_\om\rangle
		-
		\langle\Om_\om,\,\phi_\om(h)\phi_\om(g'\ot_{\rho'}(f\ot_{\rho}k))\Om_\om\rangle
		\\
		&=
		\om(h\ot (f\ot_\rho(g'\ot_{\rho'}k))) - \om(h\ot (g'\ot_{\rho'}(f\ot_\rho k)))
		\\
		&=
		\om(h\ot_\rho f\ot_\rho(g'\ot_{\rho'}k)) - \om(h\ot_{\rho'}g'\ot_{\rho'}(f\ot_\rho k))
		\\
		&=
		\om((h\ot_\rho f)\ot_\rho(g'\ot_{\rho'}k)) - \om((h\ot_{\rho'}g')\ot_{\rho'}(f\ot_\rho k))
		\\
		&=
		\om((h\ot_\rho f)\ot(g'\ot_{\rho'}k)) - \om((h\ot_{\rho'}g')\ot(f\ot_\rho k))
		\,.
	\end{align*}
	As the last expression is identical to \eqref{eq:WedgeLocalityInOmega}, the proof is finished.
\end{proof}


\section{Multiplicative deformations and their compatible states}\label{section:examples}

We now turn to the task of finding examples of deformation maps $\rho$ which meet our requirements. It will be convenient to work in momentum space most of the time, i.e., we consider the Fourier transform $f\mapsto\fti$ on $\BU$,
\begin{align*}
	\fti_n(p_1,...,p_n)
	:=
	(2\pi)^{-nd/2}\int d^dx_1\cdots d^dx_n\,f_n(x_1,...,x_n)\,e^{ip_1\cdot x_1}\cdots e^{ip_n\cdot x_n}
	\,.
\end{align*}
This map preserves the linear and product structure of $\BU$ as well as its identity element. Furthermore, the Fourier transform commutes with the action of the orthochronous Lorentz transformations, and thus $\LGpo$ acts on the momentum space wave functions in the same manner as in \eqref{eq:AlphaALa}. Translations, the ${}^*$-involution, and the reflection at the edge of the wedge take the form, $p_1,...,p_n\in\Rl^d$,
\begin{align}
	(\widetilde{\alpha_x f})_n(p_1,...,p_n)
	&=
	e^{i(p_1+...+p_n)\cdot x}\cdot \fti_n(p_1,...,p_n)
	\,,
	\label{eq:TranslationsInMomentumSpace}
	\\
	\widetilde{f^*}_n(p_1,...,p_n)
	&=
	\overline{\fti_n(-p_n,...,-p_1)}
	\label{eq:StarInMomentumSpace}
	\,,\\
	(\widetilde{\alpha_{j}f})_n(p_1,...,p_n)
	&=
	\overline{\fti_n(-j p_1,...,-j p_n)}
	\,.
	\label{eq:JInMomentumSpace}
\end{align}

After these remarks, we consider deformation maps $\rho:\BU\to\BU$ in the sense of Definition \ref{definition:DeformationMap}. In view of the structure of $\BU$, every such map is given by a family of (distributional) integral kernels $\rho_{nm}$, $n,m\in\Nl_0$, such that, $f_n\in\Ss_n$,
\begin{align}\label{eq:RhoNMKernels}
	\widetilde{\rho(f_n)}_m(p_1,...,p_m)
	&=
	\int dq_1\cdots dq_n\,\rho_{nm}(q_1,...,q_n;\,p_1,...,p_m)\, \fti_n(q_1,...,q_n)
	\,.
\end{align}
The defining properties of a deformation map restrict the possible form of the distributions $\rho_{nm}$. For example, property {\em iii)} of Definition \ref{definition:DeformationMap} requires the support of $\rho_{nm}$ to be contained in $\{(q_1,...,q_n,p_1,...,p_m)\,:\,q_1+..+q_n+p_1+..+p_m=0\}$, similar to the energy-momentum conservation of S-matrix elements.

A systematic study of deformation maps and the emerging deformed quantum field theories will be presented elsewhere. Here we consider a particularly simple class of maps $\rho:\BU\to\BU$ which preserve the grading of $\BU$ and act multiplicatively in momentum space, i.e., are of the form
\begin{align}\label{eq:MultiplicativeRho}
	\widetilde{\rho(f)}_n(p_1,...,p_n)
	=
	\rho_n(p_1,...,p_n)\cdot \fti_n(p_1,...,p_n)
	\,,\qquad
	n\in\Nl_0,\;f\in\BU\,.
\end{align}
We will refer to deformation maps of this type as {\em multiplicative deformations}. They form an abelian subgroup, denoted $\R_0$, of the group $\R$ of all deformation maps. Given $\rho\in\R_0$, the functions $\rho_n$ \eqref{eq:MultiplicativeRho} are called the {\em $n$-point functions of $\rho$}, and it is straightforward to characterize $\rho$ in terms of the $\rho_n$.

\begin{lemma}\label{lemma:MomentumSpaceDeformationMaps}
	The group $\R_0$ of multiplicative deformations of $\BU$ consists precisely of those sequences $\rho_n\in C^\infty(\Rl^{n d})$, $n\in\Nl_0$, of smooth functions which satisfy the following conditions.
	\begin{enumerate}
		\item For each multi index $\bmu\in\Nl_0^{nd}$, there exist $N_\bmu\in\Rl$ and $C_\bmu>0$ such that
		\begin{align}\label{eq:BoundOnRhoNFromAbove}
		|\partial^\bmu \rho_n(p_1,...,p_n)|
		&\leq
		C_\bmu\,(1+|p_1|^2+...+|p_n|^2)^{N_\bmu}
		\,,\qquad p_1,...,p_n\in\Rl^d\,.
	\end{align}
	\item There exists $M\in\Rl$ and $C'>0$ such that
	\begin{align}\label{eq:BoundOnRhoNFromBelow}
		|\rho_n(p_1,...,p_n)|
		&\geq
		C'\,(1+|p_1|^2+...+|p_n|^2)^{-M}
		\,,\qquad p_1,...,p_n\in\Rl^d\,.
	\end{align}
	\item For each Lorentz transformation $\La$ with $\La W_0=W_0$,
	\begin{align*}
		\rho_n(\La p_1,...,\La p_n)
		=
		\rho_n(p_1,...,p_n)
				\,,\qquad p_1,...,p_n\in\Rl^d\,.
	\end{align*}
	\item $\rho_n$ is ${}^*$-invariant,
	\begin{align*}
		\overline{\rho_n(-p_n,...,-p_1)}
		=
		\rho_n(p_1,...,p_n)
		\,,\qquad
		p_1,...,p_n\in\Rl^d\,.
	\end{align*}
	\item $\rho_0=1$.
	\end{enumerate}
\end{lemma}
\begin{proof}
	The first two conditions {\em i)}, {\em ii)} are necessary and sufficient for $\rho$ to be a homeomorphism on $\BU$: Let us first assume {\em i)}, {\em ii)} hold. Then, by condition {\em i)}, $f_n\mapsto \rho_n\cdot f_n$ maps $\Ss_n$ into $\Ss_n$. Moreover, this map is linear and it is straightforward to see that it is continuous in the Schwartz topology. By condition {\em ii)}, the $\rho_n$ are in particular non-vanishing, and the reciprocals $1/\rho_n$ are polynomially bounded by \eqref{eq:BoundOnRhoNFromBelow}. It now follows by application of the chain rule that all derivatives of $1/\rho_n$ satisfy polynomial bounds of the form \eqref{eq:BoundOnRhoNFromAbove}. Hence $f_n\mapsto f_n/\rho_n$ is also a continuous linear map from $\Ss_n$ onto $\Ss_n$, with inverse $\rho_n$.

	The map $\rho:\BU\to\BU$ on the direct sum $\BU=\bigoplus_n\Ss_n$ is continuous iff its restriction to $\Ss_n$ is continuous for each $n$ \cite{Treves:1967}. But the restriction of $\rho$ to $\Ss_n$ maps this space continuously onto $\Ss_n$, which in turn is continuously embedded in $\BU$. Hence $\rho$ is continuous, and by the same argument, one sees that $\rho^{-1}$ is continuous as well. Thus $\rho$ is a linear homeomorphism, as required in Definition \ref{definition:DeformationMap}.

	Conversely, let us now assume that $\rho$ defined as in \eqref{eq:MultiplicativeRho} is a homeomorphism of $\BU$. For such a multiplicative transformation to map $\Ss_n$ onto $\Ss_n$, it is necessary that $\rho_n$ is smooth and polynomially bounded in all derivatives, i.e., {\em i)} holds. Since $\rho^{-1}$ has the same properties, also {\em ii)} follows.

	Condition {\em iii)} is equivalent to $\alpha_\La\circ\rho=\rho\circ\alpha_\La$ for $\La$ with $\La W_0=W_0$, as a short calculation based on \eqref{eq:MultiplicativeRho} and \eqref{eq:AlphaALa} shows. Translational invariance imposes no further restrictions on $\rho$ since both $\rho$ and the translations act multiplicatively in momentum space and therefore commute automatically.

	Using \eqref{eq:MultiplicativeRho} and \eqref{eq:StarInMomentumSpace}, one easily checks that {\em iv)} is equivalent to $\rho(f^*)=\rho(f)^*$, $f\in\BU$. Since $\rho(1)_0=\rho_0$, condition {\em v)} is equivalent to $\rho(1)=1$.
\end{proof}

\noindent{\em Remark:} For $\rho\in\R_0$, the opposite deformation $\rho'=\alpha_j\circ \rho\circ\alpha_j$ is given by the $n$-point functions $\rho_n'(p_1,...,p_n)=\overline{\rho_n(-jp_1,...,-jp_n)}$. But as the Lorentz transformation $-j$ maps the wedge $W_0$ onto itself, and $-j\in\LGpo$ because $d$ is even, we can use the invariance stated in part {\em iii)} of Lemma \ref{lemma:MomentumSpaceDeformationMaps} to rewrite the $n$-point functions of the
opposite deformation as
\begin{align}\label{eq:RhoPrimeIsRhoBar}
	\rho_n'(p_1,...,p_n)
	&=
	\overline{\rho_n(p_1,...,p_n)}
	\,.
\end{align}
\vspace*{2mm}\\
The inverse $\rho^{-1}$ of a multiplicative deformation $\rho\in\R_0$ is given by the reciprocal $n$-point functions $1/\rho_n$, and thus the product $f\ot_\rho g=\rho^{-1}(\rho(f)\ot\rho(g))$ takes the following simple form in momentum space,
\begin{align}\label{eq:RhoProductKernels}
	\widetilde{(f\ot_\rho g)_n}(p_1,..,p_n)
	=
	\sum_{k=0}^n \frac{\rho_k(p_1,..,p_k)\rho_{n-k}(p_{k+1},..,p_n)}{\rho_n(p_1,..,p_n)}\fti_k(p_1,..,p_k)\gti_{n-k}(p_{k+1},..,p_n)
	\,.
\end{align}
It is clear from the conditions spelled out in Lemma \ref{lemma:MomentumSpaceDeformationMaps} that many multiplicative deformation maps exist. However, different $\rho,\hat{\rho}\in\R_0$ might induce the same product \eqref{eq:RhoProductKernels} on $\BU$. We therefore introduce an equivalence relation on $\R_0$ by defining $\rho,\hat{\rho}$ as equivalent, in symbols $\rho\sim\hat{\rho}$, if $f\ot_\rho g=f\ot_{\hat{\rho}} g$ for all $f,g\in\BU$. A multiplicative deformation $\rho\in\R_0$ is called {\em trivial} if $\rho\sim\rm id$.

\begin{lemma}
\begin{enumerate}
	\item Two deformations $\rho$, $\hat{\rho}\in\R_0$ are equivalent if and only if $\hat{\rho}\rho^{-1}$ is trivial.
	\item A deformation $\rho\in\R_0$ is trivial if and only if $\rho_n=\rho_1^{\otimes n}$, $n\in\Nl$.
	\item Let $\rho\in\R_0$. Then there exists another $\hat{\rho}\in\R_0$ with $\hat{\rho}_1=1$ and $\hat{\rho}\sim\rho$.
\end{enumerate}
\end{lemma}
\begin{proof}
	{\em i)} Assume $\rho\sim\hat{\rho}$. Then $\rho^{-1}(\rho(f)\ot\rho(g))=\hat{\rho}^{-1}(\hat{\rho}(f)\ot\hat{\rho}(g))$ for all $f,g\in\BU$, or, equivalently, $f\ot_{\hat{\rho}\rho^{-1}}g=(\rho\hat{\rho}^{-1})((\hat{\rho}\rho^{-1})(f)\ot (\hat{\rho}\rho^{-1})(g))=f\ot g$. Hence $\hat{\rho}\rho^{-1}$ is trivial. If, on the other hand, $\hat{\rho}\rho^{-1}\sim\rm id$, then $(\rho\hat{\rho}^{-1})((\hat{\rho}\rho^{-1})(f)\ot (\hat{\rho}\rho^{-1})(g))=f\ot g$, and $\rho\sim\hat{\rho}$ follows.

	{\em ii)} The triviality condition $\rho^{-1}(\rho(f)\ot\rho(g))=f\ot g$, $f,g\in\BU$, is satisfied if and only if $\rho$ is an automorphism of $\BU$. As $\rho$ is taken to be multiplicative here, it is an automorphism if and only if $\rho_n=\rho_1^{\ot n}$, $n\in\Nl$.

	{\em iii)} Let $\rho\in\R_0$. Then $\rho_1$ satisfies the conditions {\em i)--iv)} in Lemma \ref{lemma:MomentumSpaceDeformationMaps} for $n=1$, and it is easy to check that for $n\geq1$, also the functions $\sigma_n:=1/\rho_1^{\ot n}$ comply with these conditions. With $\sigma_0:=1$, this defines a multiplicative deformation $\sigma\in\R_0$ which is trivial by part {\em ii)}. According to part {\em i)}, $\hat{\rho}:=\sigma\rho$ is equivalent to $\rho$, and $\hat{\rho}_1=\rho_1/\rho_1=1$.
\end{proof}

In view of the last statement, the redundancy in describing deformed products of the form \eqref{eq:RhoProductKernels} by $n$-point functions $\rho_n$ is precisely taken into account by restricting to multiplicative deformations $\rho\in\R_0$ with trivial one point function $\rho_1=1$. We shall therefore consider only such $\rho$ in the following.
\\
\\
Following the general strategy explained in Section \ref{section:generaldeformations}, we next investigate the compatibility of $\rho\in\R_0$ with certain states $\om$ on $\BU$ (Definition \ref{definition:CompatibleState}). That is, we need to find physically relevant states such that $\om(f\ot_\rho g)=\om(f\ot g)$ for all $f,g\in\BU$. Each state on $\BU$ is given by a sequence of distributions $\om_n\in\Ss_n'$, $n\in\Nl$, its $n$-point functions, and $\om_0=1$. In momentum space, we have
\begin{align}\label{eq:om}
	\om(f)
	=
	\sum_{n=0}^\infty \int dp_1\cdots dp_n\,\omti_n(-p_1,...,-p_n)\,\fti_n(p_1,...,p_n)
	\,,\qquad
	f \in\BU\,.
\end{align}
Inserting \eqref{eq:RhoProductKernels} into the condition $\om(f\ot_\rho g)=\om(f\ot g)$, we observe that if the $\omti_n$ are measures, compatibility of $\om$ with $\rho$ is equivalent to the factorization
\begin{align}\label{eq:RhoCompatibleWithOmega-Condition}
	\rho_n(p_1,...,p_n)=\rho_k(p_1,...,p_k)\cdot\rho_{n-k}(p_{k+1},...,p_n)
	\;\text{ for all }\;(p_1,...,p_n)\in-\supp\omti_n\,,
\end{align}
for all $n,k\in\Nl_0$, $k\leq n$. For more singular distributions $\omti_n$, compatibility of $\om$ with $\rho$ poses also conditions on the derivatives of the $\rho_n$.

As the momentum space supports of $n$-point functions play a role in the compatibility question, we proceed with some remarks about relevant examples from quantum field theory. A large class of states of interest is the class of all translationally invariant states, satisfying $\om\circ\alpha_x=\om$ for all $x\in\Rl^d$. Their $n$-point functions have support at zero energy-momentum, that is,
\begin{align}\label{eq:TranslationallyInvariantSupport}
	\supp\omti_n\subset S_{\rm inv}^n := \{\bp\in\Rl^{nd}\,:\,p_1+...+p_n=0\}
	\,.
\end{align}

As vacuum states in quantum field theory, one considers the subclass of translationally invariant states satisfying the {\em spectrum condition}. These are given by $n$-point functions with \cite{Borchers:1962}
\begin{align}\label{eq:SpectrumConditionSupport}
	\supp\omti_n\subset S^n_{\rm Spec} := \{\bp\in\Rl^{nd}\,:\,p_1,p_1+p_2,...,p_1+...+p_{n-1}\in\overline{V_+},\,p_1+...+p_n=0\},
\end{align}
where $\overline{V_+}=\{q\in\Rl^d\,:\,q\cdot q\geq0,\,q^0\geq0\}$ is the closed forward light cone.

Special examples of states, related to generalized free field models, are given by {\em quasi-free} states, which are completely determined by their two-point function $\om_2$. Recall that a state $\om$ on $\BU$ is called quasi-free if
\begin{align}\label{eq:QuasiFreeNPointFunctions}
	\omti_{2n-1}
	&=
	0\,,\qquad
	\omti_{2n}(p_1,...,p_{2n})
	=
	\sum_{(\bl,\br)} \prod_{k=1}^n \omti_2(p_{l_k},p_{r_k})
	\,,\qquad n\in\Nl\,,
\end{align}
where the sum runs over all partitions $(\bl,\br)$ of $\{1,...,2n\}$ into disjoint tuples $(l_1,r_1),...,(l_n,r_n)$ with $l_k<r_k$, $k=1,...,n$. For quasi-free translationally invariant states $\om$, we have $\supp\omti_{2n-1}=\emptyset$, and
\begin{align}\label{eq:QuasiFreeSupport}
	\supp\omti_{2n}
	\subset
	S^{2n}_{\rm qf}
	:=
	\bigcup_{(\bl,\br)} \{\bp\in\Rl^{2nd}\,:\,p_{l_k}+p_{r_k}=0,\;k=1,...,n\}
	\,.
\end{align}
In view of the positivity $\om_2(f_1^*\ot f_1)\geq0$, $f_1\in\Ss_1$, we can apply Bochner's theorem to conclude that $\omti_2$ is a measure. Taking into account the special structure of the $n$-point functions \eqref{eq:QuasiFreeNPointFunctions}, it then follows that each $\omti_n$ is a measure. In particular, the compatibility of a translationally quasi-free state with a multiplicative deformation is equivalent to the factorization condition \eqref{eq:RhoCompatibleWithOmega-Condition}.

Finally, translationally quasi-free states satisfy in addition the spectrum condition if and only if  $\supp\omti_2\subset\{(p,q)\in\Rl^{2d}\,:\,p\in\overline{V_+},\,p+q=0\}$. In this last case, $\om_2$ can be represented as
\begin{align}\label{eq:KL2PointFunction}
	\omti_2(p,q)=\delta(p+q)\,w(p)\,,
\end{align}
where $w$ is a measure on $\overline{V_+}$.
\\
\\
It turns out that it is a very strong condition to require a multiplicative deformation to
be compatible with all translationally invariant states, or all translationally
invariant states satisfying the spectrum condition. In fact, since there exist
sufficiently many such states on $\BU$
\cite{Yngvason:1981}, these conditions are equivalent to requiring
\eqref{eq:RhoCompatibleWithOmega-Condition} to hold on all of
$S^n_{\rm inv}$ \eqref{eq:TranslationallyInvariantSupport}. As a necessary condition for compatibility, this yields a recursive equation
determining the $n$-point functions $\rho_n$, $n\geq2$, in terms of the two
point function $\rho_2$. In addition, several algebraic relations for
$\rho_2$ have to be satisfied for
\eqref{eq:RhoCompatibleWithOmega-Condition} to hold. One special solution, corresponding to Rieffel deformations and warped convolutions, exists, and will be recalled later on. The most general deformation two-point function complying with these conditions is presently not known, but it seems that there is little freedom for obtaining other deformations $\rho\in\R_0$ compatible with all translationally invariant states\footnote{S.~Alazzawi, work in progress.}.

Instead of asking for compatibility of $\rho$ with all translationally invariant states, we will consider in the following the less restrictive condition that $\rho$ should be compatible with all {\em quasi-free} translationally invariant states. This amounts to requiring \eqref{eq:RhoCompatibleWithOmega-Condition} to hold on the smaller domain $S^n_{\rm qf}$ \eqref{eq:QuasiFreeSupport}. We will see that an infinite family of such $\rho$ exists, providing non-trivial deformations of generalized free field theories.

Also in the case of multiplicative deformations which are compatible with
quasi-free translationally invariant states, the $n$-point functions $\rho_n$
are determined by the two-point function $\rho_2$. In the following
proposition, we show under which conditions on $\rho_2$ the required
compatibility holds. Explicit solutions of these conditions on $\rho_2$ are
then discussed in Lemma \ref{lemma:RDeformations}.

\begin{proposition}\label{proposition:DeformationsCompatibleWithQuasiFreeStates}
	Let $\rho_2\in C^\infty(\Rl^d\times\Rl^d)$ be a two-point function of a multiplicative deformation, satisfying conditions {\it i)--iv)} of Lemma \ref{lemma:MomentumSpaceDeformationMaps} for $n=2$, and in addition, $p,q\in\Rl^d$,
	\begin{align}\label{eq:QuasiFreeCompatibilityConditionsOnRho2}
		\rho_2(p,-p)
		=1
		\,,\qquad
		\rho_2(-p,q)
		=
		\rho_2(p,-q)
		=
		\rho_2(q,p)=\rho_2(p,q)^{-1}
		\,.
	\end{align}
	Define
	\begin{align}\label{eq:NPointFunctionsFromTwoPointFunction}
		\rho_0:=1
		\,,\qquad
		\rho_1(p_1):=1
		\,,\qquad
		\rho_n(p_1,...,p_n)
		:=
		\prod_{1\leq l<r\leq n}\rho_2(p_l,p_r)
		\,,\quad n\geq 2.
	\end{align}
	Then
	\begin{enumerate}
		\item The $n$-point functions \eqref{eq:NPointFunctionsFromTwoPointFunction} define a multiplicative deformation $\rho\in\R_0$.
		\item The deformed product associated with $\rho$ has the form
		\begin{align}\label{eq:DeformedProductFromTwoPointFunction}
			\widetilde{(f\ot_\rho g)}_n(p_1,..,p_n)
			=
			\sum_{k=0}^n \left(\prod_{l=1}^k \prod_{r=k+1}^n \rho_2(p_l,p_r)^{-1}\right) \fti_k(p_1,..,p_k)\,\gti_{n-k}(p_{k+1},..,p_n)
			\,.
		\end{align}
		\item Let $\om$ be a quasi-free translationally invariant state on $\BU$. Then $\rho$ \eqref{eq:NPointFunctionsFromTwoPointFunction} and $\rho^{-1}$ are compatible with $\om$, and for all $f_n\in\Ss_n$, $k\in \{0,...,n\}$, the functions
		\begin{align}\label{eq:fhat}
			\fti_{n,k,\pm}(p_1,...,p_n)
			&:=
			\fti_n(p_1,...,p_n)\cdot \prod_{l=1}^k\prod_{r=k+1}^n \rho_2(p_l,p_r)^{\pm 1}
		\end{align}
		have the same expectation value as $f_n$ in $\om$.
		\item The opposite deformation is $\rho'=\rho^{-1}$.
	\end{enumerate}
\end{proposition}
\begin{proof}
	{\em i)} We have to check that the $n$-point functions \eqref{eq:NPointFunctionsFromTwoPointFunction} satisfy the conditions of Lemma \ref{lemma:MomentumSpaceDeformationMaps}. Using the product formula \eqref{eq:NPointFunctionsFromTwoPointFunction}, it is straightforward to verify that conditions {\em i)--iii)} hold for all $n\geq2$. For {\em iv)}, we note that $\overline{\rho_2(-q,-p)}=\rho_2(p,q)$ by the ${}^*$-invariance of $\rho_2$, and compute
	\begin{align*}
		\overline{\rho_n(-p_n,...,-p_1)}
		=
		\prod_{1\leq l<r\leq n}\overline{\rho_2(-p_r,-p_l)}
		=
		\prod_{1\leq l<r\leq n}\rho_2(p_l,p_r)
		=
		\rho_n(p_1,...,p_n)
		\,.
	\end{align*}
	Hence {\em iv)} holds, and by definition \eqref{eq:NPointFunctionsFromTwoPointFunction}, also {\em v)} is satisfied. Thus $\rho\in\R_0$.

	{\em ii)} Here we just have to insert the definition of $\rho_n$ \eqref{eq:NPointFunctionsFromTwoPointFunction} into \eqref{eq:RhoProductKernels}. Let $n\in\Nl_0$ and $k\in\{0,...,n\}$. Then
	\begin{align}
		\frac{\rho_k(p_1,...,p_k)\rho_{n-k}(p_{k+1},...,p_n)}{\rho_n(p_1,...,p_n)}
		&=
		\frac{\prod\limits_{1\leq l'<r'\leq k} \rho_2(p_{l'},p_{r'})\cdot \prod\limits_{k+1\leq l''<r''\leq n} \rho_2(p_{l''},p_{r''})}{\prod\limits_{1\leq l<r\leq n} \rho_2(p_l,p_r)}
		\nonumber
		\\
		&=
		\prod_{l=1}^k \prod_{r=k+1}^n \rho_2(p_l,p_r)^{-1}
		\,,
		\label{eq:ProductKernelsFromTwoPointFunction}
	\end{align}
	and \eqref{eq:DeformedProductFromTwoPointFunction} follows.

	{\em iii)} For $\rho$ to be compatible with all quasi-free translationally invariant states, we will show that \eqref{eq:ProductKernelsFromTwoPointFunction} equals 1 for even $n=2N$, $k\in\{0,...,2N\}$, and $\bp\in -S^{2N}_{\rm qf}$ \eqref{eq:QuasiFreeSupport}. Fixing such $N,k$, let $(\bl,\br)=\{(l_1,r_1),...,(l_N,r_N)\}$ be a partition of $\{1,...,2N\}$ into pairs $(l_j,r_j)$ as in \eqref{eq:QuasiFreeSupport}, and $\bp\in\Rl^{2Nd}$ with $p_{l_j}+p_{r_j}=0$, $j=1,...,N$. We split the partition into three parts: First, the pairs $(l_j,r_j)$ with $l_j,r_j\leq k$, denoted $\{(\lhat_1,\rhat_1),...(\lhat_L,\rhat_L)\}$, second, the pairs $(l_j,r_j)$ with $l_j\leq k<r_j$, denoted $\{(\lti_1,\rti_1),...,(\lti_M,\rti_M)\}$, and third, the pairs $(l_j,r_j)$ with $k<l_j,r_j$, denoted $\{(\lcheck_1,\rcheck_1),...,(\lcheck_R,\rcheck_R)\}$. Clearly, these sets are disjoint, and their union is $\{(l_1,r_1),...,(l_N,r_N)\}$, i.e. in particular $\{\lhat_1,...,\lhat_L,\rhat_1,...,\rhat_L,\lti_1,...,\lti_M\}=\{1,...,k\}$ and $\{\rti_1,...,\rti_M,\lcheck_1,...,\lcheck_R,\rcheck_1,...,\rcheck_R\}=\{k+1,...,2N\}$.

	We now rewrite \eqref{eq:ProductKernelsFromTwoPointFunction} using this splitting as well as the support condition $p_{l_j}+p_{r_j}=0$ and the properties \eqref{eq:QuasiFreeCompatibilityConditionsOnRho2}, which give
	\begin{align*}
		\prod_{l=1}^k\prod_{r=k+1}^{2N} \rho_2(p_l,p_r)^{-1}
 		&=
		\prod_{r=k+1}^{2N}\left(
		\prod_{i=1}^L   (\rho_2(p_{\lhat_i},p_r)^{-1}\rho_2(p_{\rhat_i},p_r)^{-1})
		\cdot
  		\prod_{j=1}^M \rho_2(p_{\lti_j},p_r)^{-1}
  		\right)
  		\\
  		&=
		\prod_{r=k+1}^{2N}\left(
		\prod_{i=1}^L
   		( \rho_2(p_{\lhat_i},p_r)^{-1} \rho_2(-p_{\lhat_i},p_r)^{-1})
  		\cdot
  		\prod_{j=1}^M  \rho_2(p_{\lti_j},p_r)^{-1}
  		\right)
  		\\
  		&=
    		\prod_{j=1}^M
  		\prod_{r=k+1}^{2N}  \rho_2(p_{\lti_j},p_r)^{-1}
  		\\
  		&=
		\prod_{j=1}^M
  		\left(
  		\prod_{t=1}^R( \rho_2(p_{\lti_j},p_{\lcheck_t})^{-1} \rho_2(p_{\lti_j},p_{\rcheck_t})^{-1})
  		\prod_{i=1}^M  \rho_2(p_{\lti_j},p_{\rti_i})^{-1}
  		\right)
  		\\
  		&=
		\prod_{j=1}^M
  		\left(
  		\prod_{t=1}^R( \rho_2(p_{\lti_j},p_{\lcheck_t})^{-1} \rho_2(p_{\lti_j},-p_{\lcheck_t})^{-1})
		\prod_{i=1}^M  \rho_2(p_{\lti_j},p_{\rti_i})^{-1}
  		\right)
  		\\
		&=
  		\prod_{j=1}^M
		\prod_{i=1}^M  \rho_2(p_{\lti_j},p_{\lti_i})
  		\,.
	\end{align*}
	Using the ${}^*$-invariance of $\rho_2$, and \eqref{eq:QuasiFreeCompatibilityConditionsOnRho2}, we get $\rho_2(p,p)=\overline{\rho_2(-p,-p)}=\overline{\rho_2(-p,p)}^{-1}=1$.  Hence in the product $\prod_{i,j}^M  \rho_2(p_{\lti_j},p_{\lti_i})$ the diagonal terms  $\rho_2(p_{\lti_j},p_{\lti_j})=1$ drop out. The off-diagonal terms appear in reciprocal pairs $\rho_2(p_{\lti_i},p_{\lti_j})$ and $\rho_2(p_{\lti_j},p_{\lti_i})=\rho_2(p_{\lti_i},p_{\lti_j})^{-1}$, and therefore drop out as well. As the partition $(\bl,\br)$ was arbitrary, the compatibility of $\rho$ and $\om$ follows.

	Replacing $\rho_2$ by $1/\rho_2$, we also have compatibility of $\rho^{-1}$ and $\om$. The equation $\om(f_{n,k,\pm})=\om(f_n)$ is just a reformulation of these compatibility statements.

	{\em iv)} As $\rho_2$ satisfies \eqref{eq:QuasiFreeCompatibilityConditionsOnRho2}, and is invariant under the ${}^*$-operation \eqref{eq:StarInMomentumSpace}, we have
	\begin{align*}
		\overline{\rho_2(p,q)}=\rho_2(-q,-p)=\rho_2(q,p)=\rho_2(p,q)^{-1}
		\,.
	\end{align*}
	In view of the product form of the $\rho_n$, this implies $\overline{\rho_n}=1/\rho_n$. But the $n$-point functions of the opposite deformation $\rho'$ are the conjugates of the $\rho_n$ \eqref{eq:RhoPrimeIsRhoBar}. Hence $\rho'=\rho^{-1}$.
\end{proof}

Having reduced the problem of finding deformations compatible with quasi-free translationally invariant states to conditions on the two-point function $\rho_2$, we next solve these conditions by discussing suitable two-point functions. These will be realized in terms of a deformation function $R$ and an admissible matrix $Q$, defined below.

\begin{definition}\label{definition:DeformationFunction}
	A deformation function is a smooth function $R:\Rl\to\Cl$ such that
	\begin{enumerate}
		\item
		\begin{align}\label{eq:R-Relations}
			R(a)^{-1}
			&=
			\overline{R(a)}
			=
			R(-a)\,,
			\qquad
			R(0)=1
			\,,
		\end{align}
		\item For each $k\in\Nl$, there exists $C_k, N_k>0$, such that
			\begin{align}\label{eq:BoundsOnR}
				\left|\frac{\partial^k R(a)}{\partial a^k}\right|
				&
				\leq
				C_k\,(1+a^2)^{N_k}
				\,,\qquad
				a\in\Rl\,,
			\end{align}
		\item The Fourier transform $\Rti\in\Ss_1'$ of $R$ has support on the positive half line.
	\end{enumerate}
\end{definition}

Note that the support restriction on $\Rti$ amounts to requiring that $R$ has an analytic continuation to the upper half plane. More precisely \cite[Thm. IX.16]{ReedSimon:1975}, any deformation function is the boundary value in the sense of $\Ss_1'$ of a function analytic in the upper half plane, satisfying polynomial bounds at infinity and at the real boundary. Conversely, if $R$ is a function analytic on the upper half plane, satisfying suitable polynomial bounds, then its distributional boundary value on the real line exists, and its Fourier transform has support in the right half line. As concrete examples, consider the functions
\begin{align}\label{eq:RExampleFunctions}
	R(a)
	=
	e^{ica}\prod_{k=1}^N \frac{z_k-a}{z_k+a}
	\,,\qquad
	c\geq 0,\;\im z_1,...,\im z_N\geq 0
	\,,
\end{align}
where with each $z_k$, also $-\overline{z_k}$ is contained in the set of zeros $\{z_1,...,z_N\}$. As these functions satisfy the first two conditions of Definition \ref{definition:DeformationFunction}, and furthermore have bounded analytic continuations to the upper half plane, they are examples of deformation functions.

\begin{lemma}\label{lemma:RDeformations}
	Let $R$ be a deformation function, and let $Q$ be a $(d\times d)$-matrix which is antisymmetric w.r.t. the Minkowski inner product on $\Rl^d$, and satisfies
	\begin{align}\label{eq:QLambdaInvariance}
		\La Q\La^{-1}
		&=
		\left\{
		\begin{array}{rcl}
			Q &;& \La\in\LGpo\;\text{with}\;\La W_0=W_0\\
			-Q &;& \La\in\LG_+^\downarrow\;\text{with}\;\La W_0=W_0
		\end{array}
		\right.
	\end{align}
	Then the deformation two-point function
	\begin{align}\label{eq:Rho2FromRAndQ}
		\rho_2(p,q)
		:=
		R(-p\cdot Qq)
	\end{align}
	satisfies all assumptions of Proposition \ref{proposition:DeformationsCompatibleWithQuasiFreeStates}.
\end{lemma}
\begin{proof}
	Checking the conditions {\em i)--iv)} of Lemma \ref{lemma:MomentumSpaceDeformationMaps}, and the additional properties \eqref{eq:QuasiFreeCompatibilityConditionsOnRho2} is a matter of straightforward computation making use of the listed properties of $R$, and the antisymmetry and partial Lorentz invariance of $Q$. Note that the minus sign for time-reversing Lorentz transformations which appears in \eqref{eq:QLambdaInvariance} cancels against the complex conjugation of $\alpha_\La$, $\La\in\LG_+^\downarrow$, since $\overline{R(-a)}=R(a)$.
\end{proof}

Given a deformation map $\rho=\rho(R,Q)$ of the form described above, it is straightforward to check that a general Poincar\'e transformation $(x,\La)\in\PG_+$ acts on the associated deformed product according to
\begin{align}\label{eq:PoincareTransformsOfRQProduct}
	\alpha_{x,\La}(f\ot_{\rho(R,Q)}g)
	=
	\alpha_{x,\La}(f)\ot_{\rho(R,\pm\La Q\La^{-1})}\alpha_{x,\La}(g)
	\,,\qquad
	f,g\in\BU\,,
\end{align}
where the $\pm$-sign is ``$+$'' for orthochronous and ``$-$'' for non-orthochronous $\La$. This identity can easily be verified on the basis of \eqref{eq:DeformedProductFromTwoPointFunction} and \eqref{eq:R-Relations}.

It has been shown in \cite{GrosseLechner:2007} that the most general matrix satisfying \eqref{eq:QLambdaInvariance} is, in case the spacetime dimension is $d=4$ or $d\neq4$,
\begin{align}\label{eq:CovariantQ}
	Q
	&=
	\left(
		\begin{array}{cccc}
			0&\kappa&0&0\\
			\kappa&0&0&0\\
			0&0&0&\kappa'\\
			0&0&-\kappa'&0
		\end{array}
	\right)
	\,,\qquad
	Q
	=
	\left(
	\begin{array}{ccccc}
		0 & \kappa & 0 & \cdots & 0\\
		\kappa & 0  & 0 & \cdots & 0\\
		0      & 0  & 0 & \cdots & 0\\
		\vdots & \vdots &\vdots&\ddots&\vdots\\
		0 & 0&0&\cdots&0
	\end{array}
	\right)\,,
\end{align}
with arbitrary parameters $\kappa,\kappa'\in\Rl$. For Lorentz transformations $\La$ which map the wedge $W_0$ onto its causal complement $W_0'$, one has
\begin{align}\label{eq:QLambdaInvariance2}
	\La Q\La^{-1}
	&=
	\left\{
	\begin{array}{rcl}
		-Q &;& \La\in\LGpo\;\text{with}\;\La W_0=W_0'\\
		Q &;& \La\in\LG_+^\downarrow\;\text{with}\;\La W_0=W_0'
	\end{array}
	\right.\;.
\end{align}
This implies that for fixed $R$, the opposite deformation is given by inverting the sign of $Q$,
\begin{align}\label{eq:ReflectionJOfRQProduct}
	\alpha_j(f\ot_{\rho(R,Q)}g)
	=
	\alpha_j(f)\ot_{\rho(R,-Q)}\alpha_j(g)
	\,,\qquad
	f,g\in\BU\,.
\end{align}

We also mention that the deformations $\rho(R,Q)$ naturally lead to one-parameter families of deformation maps $\rho(R,\la\cdot Q)$, $\la\in\Rl$. In the limit $\la\to0$, we recover the undeformed product.

\begin{proposition}\label{proposition:RQProductContinuousInParameter}
	Let $R$ be a deformation function and $Q$ a real $(d\times d)$-matrix. Then, for all $f,g\in\BU$,
	\begin{align}\label{eq:LimitOfLambdaProduct}
		\lim_{\la\to0}f\ot_{\rho(R,\la\cdot Q)}g
		=
		f\ot g\,.
	\end{align}
\end{proposition}
\begin{proof}
	As $R(0)=1$, the functions $r_\la(p_1,...,p_n):=\prod_{l,r} R(\la\,p_l\cdot Qp_r)$ \eqref{eq:DeformedProductFromTwoPointFunction} appearing in the product $\ot_{\rho(R,\la\cdot Q)}$ converge pointwise to the constant function $1$ as $\la\to0$. This limit is also valid in a stronger topology: Making use of the polynomial boundedness of the derivatives of $R$, it is not difficult to show that for any multi index $\boldsymbol\mu\in\Nl^{nd}$, there exists $N(\boldsymbol\mu)\in\Rl$ such that
	\begin{align*}
		\lim_{\la\to0}\sup_{\bp\in\Rl^{nd}}\frac{|\partial_\bp^{\boldsymbol\mu} (r_\la(\bp)-1)|}{(1+\|\bp\|^2)^{N(\boldsymbol\mu)}}
		=
		0\,.
	\end{align*}
	It then follows by straightforward estimates that $(f\ot_{\rho(R,\la\cdot Q)}g)_n=r_\la\cdot (f\ot g)_n\to (f\ot g)_n$ as $\la\to0$, in the topology of $\Ss_n$, $n\in\Nl$. This implies the claimed limit \eqref{eq:LimitOfLambdaProduct}.
\end{proof}

The simplest non-trivial deformation function is $R(a):=e^{ia}$. This example was studied in \cite{GrosseLechner:2008}. For this function, the corresponding deformed product
\begin{align*}
	\widetilde{(f\otimes_\rho g)}_n(p_1,...,p_n)
	&=
	\sum_{k=0}^n e^{i(p_1+...+p_k)\cdot Q(p_{k+1}+...+p_n)}\,\fti_k(p_1,...,p_k)\,\gti_{n-k}(p_{k+1},...,p_n)
\end{align*}
can also be written as
\begin{align*}
	(f\otimes_\rho g)_n(x_1,...,x_n)
	=
	(2\pi)^{-d}\int dq\,dy\,e^{-iq\cdot y}\,(\alpha_{Qq}f \ot \alpha_y g)_n(x_1,...,x_n)
	\,.
\end{align*}
It is thus identical to the Rieffel-deformation \cite{Rieffel:1992} of the tensor
product $\otimes$ with the $\Rl^d$-action $\alpha|_{\Rl^d}$
\cite{GrosseLechner:2008}. In particular, it follows that all translationally
invariant states are compatible with this deformation, a fact proven in a
$C^*$-framework in \cite{Rieffel:1993}. In our present setting, we immediately
see that because of the antisymmetry of $Q$, we have $e^{-i (p_1+...+p_k)\cdot
Q(p_{k+1}+...+p_n)}=1$ for all $\bp\in S^n_{\rm inv}$
\eqref{eq:TranslationallyInvariantSupport}. By Proposition \ref{proposition:RhoOmCompatibility}, this implies
the compatibility of this deformation with all translationally invariant (not
necessarily quasi-free) states.

However, for the deformations given by a general deformation function $R$, the restriction
to quasi-free states is necessary. In fact, assume that the $n$-point functions $\rho_n$ defined in terms of $R$ by \eqref{eq:Rho2FromRAndQ} and \eqref{eq:NPointFunctionsFromTwoPointFunction} satisfy \eqref{eq:RhoCompatibleWithOmega-Condition} on $S^n_{\rm inv}$ \eqref{eq:TranslationallyInvariantSupport}. Taking $n=4$ and $k=1$ in
\eqref{eq:RhoCompatibleWithOmega-Condition}, we then have $R(p_1\cdot
Qp_2)R(p_1\cdot Qp_3)R(p_1\cdot Qp_4)=1$ for all $p$ with $p_1+...+p_4=0$.
Inserting $p_4=-(p_1+p_2+p_3)$ and making use of the antisymmetry of $Q$ as well
as \eqref{eq:R-Relations} yields the condition
\begin{align*}
	R(p_1\cdot Qp_2)R(p_1\cdot Qp_3)
	&=
	R(p_1\cdot Qp_2+p_1\cdot Qp_3)
	\,.
\end{align*}
Since $p_1,p_2,p_3$ can be chosen independently, this condition is only
satisfied for $R(a)=e^{ica}$. We will therefore restrict our attention to
quasi-free translationally invariant states $\om$ in the following.
\\
\\
We now consider the question under which conditions a multiplicative deformation $\rho$ is wedge-local in the GNS representation associated with a quasi-free translationally invariant state $\om$. As we have been working with the full tensor algebra $\BU$ instead of its quotient $\BU/\Le$ by the locality
ideal $\Le$, we have to consider states annihilating $\Le$. Picking such a state $\om$, we recall from Lemma \ref{lemma:LocalityInOmega} that wedge-locality in the GNS representation of
$(\BU,\om)$ amounts to
\begin{align}\label{eq:WedgeLocalityInOmega-Again}
	\om((u\ot_\rho f)\ot (g'\ot_{\rho'} v))
	=
	\om((u\ot_{\rho'}g')\ot (f\ot_{\rho} v))
\end{align}
for all $f\in\BU(W_0)$, $g'\in\alpha_j(\BU(W_0))=\BU(W_0')$, and all $u,v\in\BU$.

To motivate the following steps, it is instructive to recall the known results about the special case $R(a)=e^{ia}$ first. In this context, \eqref{eq:WedgeLocalityInOmega-Again} is known to hold for a translationally invariant state $\om$ annihilating $\Le$ if $\om$ satisfies also the spectrum condition and $Q$ is {\em admissible} in the sense that $Q\overline{V_+}\subset\overline{W_0}$ \cite{BuchholzLechnerSummers:2010}. This interplay of locality and spectral properties can be understood as follows. The spectrum condition restricts the supports of the $n$-point functions $\omti_n$ to those $\bp\in\Rl^{nd}$ with $p_1,p_1+p_2,...,p_1+...+p_n\in\overline{V_+}$ \eqref{eq:SpectrumConditionSupport}. This implies that in \eqref{eq:WedgeLocalityInOmega-Again}, we may restrict to $u$ with $\supp\uti\subset -S^n_{\rm inv}$. For those $u$, in the deformed product
\begin{align*}
	\widetilde{(u\ot_\rho f)}_n(p_1,...,p_n)
	&=
	\sum_{k=0}^n e^{i(p_1+...+p_k)\cdot
	Q(p_{k+1}+...+p_n)}\,\tilde{u}_k(p_1,...,p_k)\,\fti_{n-k}(p_{k+1},...,p_n)
	\\
	&=
	\sum_{k=0}^n
\tilde{u}_k(p_1,...,p_k)\,\alpha_{-Q(p_1+...+p_k)}(\fti)_{n-k}(p_{k+1},...,p_n)
\end{align*}
only translations of $f$ in the directions $-Q(p_1+...+p_k)$ appear, which by admissibility of $Q$ lie in the right wedge $\overline{W_0}$. But translations along $x\in\overline{W_0}$ preserve the support of $f\in\BU(W_0)$. Similar
arguments can be applied to the other terms in \eqref{eq:WedgeLocalityInOmega-Again}, showing that $g'\in\BU(W_0')$ is effectively translated in the opposite direction, so that also the support of $g'$ in $W_0'=-W_0$ is preserved. Thus the effect of the deformation consists in shifting the spacelike supports of $f$ and $g'$ apart, and the locality condition of $\om$ then allows to conclude that \eqref{eq:WedgeLocalityInOmega-Again} holds \cite{GrosseLechner:2008}.

As we are working here with a family of deformations containing $R(a)=e^{ia}$, we will in the following also require that $\om$ satisfies the spectrum condition and $Q$ is admissible. This last condition simply amounts to choosing the parameter $\kappa$ appearing in $Q$ \eqref{eq:CovariantQ} non-negative.

The multiplicative deformations given by a function $R$ which are not of exponential form do not simply act as translations on $\BU$, and the preceding locality argument for the case $R(a)=e^{ia}$ has to be adapted. Here the half-sided support of the Fourier transform $\Rti$ (Definition \ref{definition:DeformationFunction} {\em iii)}) comes into play, which makes it possible to control the effect of the deformed products $\ot_\rho$, $\ot_{\rho'}$ on the spacetime supports of suitable test functions.

\begin{proposition}\label{proposition:RShiftOperator}
Let $R$ be a deformation function and define, $x\in\Rl^d$,
	\begin{align}\label{eq:tauR}
		\tau^R_x&:\BU\to\BU,
		\nonumber
		\\
		\widetilde{(\tau^R_x f)_n}(p_1,...,p_n)
		&:=
		\fti_n(p_1,...,p_n)\cdot\prod_{k=1}^n R(x\cdot p_k)
		\,.
	\end{align}
\begin{propositionlist}
	\item $\tau_x^R$ is a continuous automorphism of $\BU$ for any $x\in\Rl^d$. For $x\in\overline{W_0}$, one has
		\begin{align}
			\tau^R_{\pm x}(\BU(\pm W_0))\subset\BU(\pm W_0)
			\,.
		\end{align}
	\item Let $n,m\in\Nl_0$, $h^\pm\in\Ss_m$ with $\supp\hti^\pm\subset\overline{V_\pm}$, and $f\in\Ss_n(W_0)$, $g'\in\BU(W_0')$. Then the deformation map $\rho$ given by $R$ and an admissible matrix $Q$ satisfies
	\begin{align}
		\supp (h^-\ot_\rho f) &\subset \Rl^{md} \times (W_0)^{\times n}
		\,,\\
		\supp(h^-\ot_{\rho'}g') &\subset \Rl^{md} \times (W_0')^{\times n}
		\,,\\
		\supp(g'\ot_{\rho'} h^+) &\subset (W_0')^{\times n}\times\Rl^{md}
		\,,\\
		\supp(f\ot_{\rho} h^+) &\subset (W_0)^{\times n}\times\Rl^{md}
		\,.
	\end{align}
\end{propositionlist}
\end{proposition}
\begin{proof}
	{\em i)} The linearity and continuity of each $\tau_x^R$ is clear. In momentum space, $\tau_x^R$ multiplies by the tensor product function $R_x^{\ot n}$, $R_x(p):=R(x\cdot p)$. Hence $\tau_x^R$ is an algebra homomorphism. It is also invertible, with inverse $(\tau_x^R)^{-1}=\tau_{-x}^R$, because $R(-a)=R(a)^{-1}$. By definition, $\tau_x^R$ has the identity of $\BU$ as a fixed point, and since $\overline{R(-a)}=R(a)$, we also have $\tau_x^R (f)^*=\tau_x^R(f^*)$, $f\in\BU$. This shows that $\tau_x^R$ is an automorphism of $\BU$. Note that although $\tau_0^R=\rm id$ and $(\tau_x^R)^{-1}=\tau_{-x}^R$, the group law $\tau_x^R\tau_y^R=\tau_{x+y}^R$ holds only if $R$ is of exponential form.

	To check the claims about the supports in wedges, let $f\in\BU(W_0)$, {\em i.e.}, $\supp f_n\subset W_0^{\times n}$ for all $n\in\Nl$. The function $R_x^{\ot n}$ defines a tempered distribution, and therefore has a Fourier transform $\Rti_x^{\ot n}$ in $\Ss_n'$ such that $\tau_x^R(f_n)=(2\pi)^{-nd/2}\,\Rti_x^{\ot n}*f_n$. Explicitly, $\by,\bz,\bp\in\Rl^{nd}$,
	\begin{align*}
		(\tau^R_x f)_n(\by)
		&=
		(2\pi)^{-nd/2}
		\int d\bz f_n(\by-\bz)
		\int d\bp \prod_{k=1}^n\left(e^{-ip_k\cdot z_k}\,R(x\cdot p_k)\right)
		\\
		&=
		(2\pi)^{-nd/2}
		\int d\bz f_n(\by-\bz)
		\int d\bp \prod_{k=1}^n\left(e^{-ip_k\cdot z_k}\,(2\pi)^{-1/2}\int d\la_k \,e^{i\la_k\,(x\cdot p_k)}\Rti(\la_k)\right)
		\\
		&=
		(2\pi)^{+n(d-1)/2}
		\int d\la_1\cdots d\la_n\,\Rti(\la_1)\cdots\Rti(\la_n)
		\int d\bz f_n(\by-\bz)
		\prod_{k=1}^n \delta(z_k-\la_k\,x)
		\\
		&=
		(2\pi)^{+n(d-1)/2}
		\int d\la_1\cdots d\la_n\,\Rti(\la_1)\cdots\Rti(\la_n)
		f_n(y_1-\la_1x,...,y_n-\la_n x)
		\,.
	\end{align*}
	The wedge $W_0$ has the two geometric properties $\la W_0 \subset\overline{W_0}$ for $\la\geq0$ and $W_0+\overline{W_0}\subset W_0$. Since $\supp\Rti\subset\Rl_+$, all $\la_k$ appearing in this integral are positive, and since $x \in\overline{W_0}$, we have $\la_1 x,...,\la_n x \in\overline{W_0}$. Taking into account that the support of $f_n$ lies in $W_0^{\times n}$, we find $\supp (\tau_x^R f)_n \subset \supp f_n + \overline{W_0}^{\times n}  \subset W_0^{\times n}+\overline{W_0}^{\times n} \subset W_0^{\times n}$, and hence $\tau_x^R\BU(W_0)\subset \BU(W_0)$. 	The arguments leading to $\tau_{-x}^R(\BU(-W_0))\subset\BU(-W_0)$ are completely analogous.

	{\em ii)} The deformed product $h^-\ot_{\rho}f$  can be expressed with the shift automorphisms $\tau^R_x$ \eqref{eq:tauR} as, $\bk\in\Rl^{md}$, $\bp\in\Rl^{nd}$,
	\begin{align*}
		(h^-\ot_{\rho}f)^{\sim}(\bk,\bp)
		&=
		\hti^-(\bk)
		\prod_{l=1}^m\prod_{r=1}^n R(k_l\cdot Qp_r) \cdot \fti(\bp)
		\\
		&=
		\hti^-(\bk)
		\left( \tau^R_{-Qk_1}\cdots \tau^R_{-Qk_m}f\right)^{\sim}(\bp)
		\,.
	\end{align*}
	As $\supp \hti^-\subset\overline{V_-}$ and $Q$ is admissible, the vectors $-Qk_1$, ..., $-Qk_n$ all lie in $\overline{W_0}$, and by part a), we have $\tau^R_{-Qk_1}\cdots \tau^R_{-Qk_m}f\in\Ss_n(W_0)$. Hence $h^-\ot_{\rho}f$ has support in $\Rl^{md}\times W_0^{\times n}$.

	In comparison, in $(g'\ot_{\rho'}h^+)$, the support of $g'$ lies in $-W_0$ instead of $W_0$, and $\rho$ is replaced by $\rho'$. But the opposite deformation is given by the same deformation function $R$, and matrix $-Q$ instead of $Q$. Hence we can repeat the above argument with shifts $+Qk_1,...,+Qk_m\in-{\overline{W_0}}$, preserving the support of $g'$ in $-W_0$, i.e. $\supp (g'\ot_{\rho'}h^+)\subset(W_0')^{\times n}\times\Rl^{md}$.

	The third and fourth function can be rewritten as
	\begin{align*}
		(g'\ot_{\rho'}h^+)=((h^+)^* \ot_{\rho'} (g')^*)^*
		\,,\qquad
		(f\ot_{\rho} h^+)=((h^+)^*\ot_{\rho} f^*)^*
		\,.
	\end{align*}
	As the ${}^*$-involution preserves supports in spacetime, but reflects supports in momentum space about the origin, we have $\supp(\hti^+)^*\subset \overline{V_-}$, $\supp (g')^*\subset (-W_0)$, $\supp f^*\subset W_0$, as in the first two functions $h^-\ot_{\rho}f$, $h^-\ot_{\rho'}g'$. Taking also into account $(\Ss_m\ot\Ss_n(\pm W_0))^*=\Ss_n(\pm W_0)\ot\Ss_m$, the claim about the supports of $(g'\ot_{\rho'}h^+)$ and $(f\ot_{\rho} h^+)$ follows.
\end{proof}

As a preparation for the wedge-locality proof, we recall two facts about quasi-free states satisfying the spectrum condition respectively states vanishing on the locality ideal.

\begin{lemma}\label{lemma:QuasiFreeGNS}
	Let $\om$ be a quasi-free translationally invariant state on $\BU$	which satisfies the spectrum condition, and consider some $f\in\BU$ and the vector $\Psi_\om(f)$ representing $f$ in the GNS representation space of $(\BU,\om)$. If
	\begin{align*}
		\om(h^-\ot f)=0
		\quad\text{for all }\;h^-\in\BU\;\text{ with }\;\supp \hti^-\subset\overline{V_-}
		\,,
	\end{align*}
	then $\Psi_\om(f)=0$.
\end{lemma}
\begin{proof}
	By the GNS construction, we have $\om(h^-\ot f)=\langle \Psi_\om((h^-)^*),\,\Psi_\om(f)\rangle$, where the momentum space support of $(h^-)^*$ is $\supp(\hti^-)^*\subset\overline{V^+}$. We thus have to show that the space $\DD_+\subset\Hil_\om$ of all $\Psi_\om(h^+)$, where $\supp\hti^+\subset\overline{V_+}$, is dense. This is a consequence of $\om$ being quasi-free and satisfying the spectrum condition. In fact, in this situation, $\om_2$ has the form \eqref{eq:KL2PointFunction} with a measure $w$ on $\overline{V_+}$, and the GNS representation space $\Hil_\om$ is the Bose Fock space over the single particle space $L^2(\overline{V_+},w(p)dp)$. For functions $h_n^+\in\Ss_n$ whose support in momentum space does not intersect the backward lightcone, $\Psi_\om(h_n^+)$ is a vector in the $n$-particle space $L^2(\overline{V_+},w(p)dp)^{\ot_{\rm sym} n}$, given by symmetrization in all variables of the Fourier transform $\hti_n^+$. The $n$-particle vectors obtained in this manner form a dense subspace of the $n$-particle space. Since we can take arbitrary $n$, the density of $\DD_+$ follows.
\end{proof}

\begin{lemma}\label{lemma:LocalityAndTheFlip}
	Let $F\in\Ss_{n+m}$, $G\in\Ss_{n'+m'}$, $n,m,n',m'\in\Nl_0$, such that $\supp F\subset\Rl^{md}\times (W_0)^{\times n}$ and $\supp G\subset (W_0')^{\times n'}\times\Rl^{m'd}$. Let $\tau:\Ss_{m+n+n'+m'}\to\Ss_{m+n+n'+m'}$ denote the flip
	\begin{align}\label{eq:FlipTau}
		(\tau H)(\by,\bx,\bx',\by')
		:=
		H(\by,\bx',\bx,\by')
		\,,\qquad
		\by\in\Rl^{md},\,\bx\in\Rl^{nd},\,\bx'\in\Rl^{n'd},\,\by'\in\Rl^{m'd}\,.
	\end{align}
	Then for each state $\om$ on $\BU$ which annihilates the locality ideal, we have
	\begin{align}\label{eq:FlipInLocalState}
		\om(F\ot G)=\om(\tau(F\ot G))
		\,.
	\end{align}

\end{lemma}
\begin{proof}
	In view of the support properties of $F$ and $G$, we can represent these functions as $F=\sum_{t=1}^\infty l^{(t)}\ot r^{(t)}$,  $G=\sum_{t=1}^\infty a^{(t)}\ot b^{(t)}$, with $l^{(t)}\in\Ss_m$, $r^{(t)}\in\Ss_n(W_0)$, $a^{(t)}\in\Ss_{n'}(W_0')$, $b^{(t)}\in\Ss_{m'}$, and these series converge in the topology of $\BU$. Because the supports of the $r^{(t)}$ and $a^{(t)}$ are spacelike separated and $\om$ annihilates the locality ideal, we have
	\begin{align*}
		\sum_{t=1}^T\sum_{s=1}^S \om(l^{(t)}\ot r^{(t)} \ot a^{(s)}\ot b^{(s)})
		&=
		\sum_{t=1}^T\sum_{s=1}^S \om(l^{(t)}\ot a^{(s)}\ot r^{(t)} \ot b^{(s)})
		\\
		&=
		 \om(\tau(\sum_{t=1}^T\sum_{s=1}^S\;l^{(t)}\ot r^{(t)} \ot a^{(s)}\ot b^{(s)}))\,.
	\end{align*}
	Making use of the continuity of $\tau$ and $\om$, the equality \eqref{eq:FlipInLocalState} follows from the above calculation in the limit $T\to\infty$, $S\to\infty$.
\end{proof}

\begin{theorem}\label{theorem:Locality}
	Let $R$ be a deformation function, and let $Q$ be an admissible matrix. Then the deformation $\rho$ given by $R$ and $Q$ via \eqref{eq:Rho2FromRAndQ} and \eqref{eq:NPointFunctionsFromTwoPointFunction} is wedge-local in any quasi-free translationally invariant state $\om$ which satisfies the spectrum condition and vanishes on the locality ideal.
\end{theorem}
\begin{proof}
	Let $n,n'\in\Nl_0$, and $f\in\Ss_n(W_0)$, $g'\in\Ss_{n'}(W_0')$. We have to show that $\phi^\rho_\om(f)$ and $\phi^{\rho'}_\om(g')$ commute on the Wightman domain $\phi_\om(\BU)\Om_\om$ in the GNS space $\Hil_\om$. In view of Lemma \ref{lemma:LocalityInOmega}, this is equivalent to showing that for arbitrary $h^+\in\Ss_m,h^-\in\Ss_{m'}$, $m,m'\in\Nl_0$, one of the following equivalent equations holds,
	\begin{align}
		\om(h^-\ot(f\ot_\rho(g'\ot_{\rho'}h^+)))
		&=
		\om(h^-\ot(g'\ot_{\rho'}(f\ot_{\rho}h^+)))
		\nonumber
		\\
		\Longleftrightarrow
		\om((h^-\ot_\rho f)\ot (g'\ot_{\rho'}h^+))
		&=
		\om((h^-\ot_{\rho'}g')\ot (f\ot_{\rho}h^+))
		\label{eq:LocalityInOmega}
		\\
		\Longleftrightarrow
		\om(((h^-\ot_\rho f)\ot_{\rho'} g')\ot h^+)
		&=
		\om(((h^-\ot_{\rho'}g')\ot_\rho f)\ot h^+)
		\nonumber
		\\
		\Longleftrightarrow
		\om((h^+)^*\ot((g')^*\ot_{\rho'}(f^*\ot_\rho (h^-)^*)))
		&=
		\om((h^+)^*\ot(f^*\ot_{\rho}((g')^*\ot_{\rho'} (h^-)^*)))
		\,.
		\nonumber
	\end{align}
	Considering the first equation, we note that by Lemma \ref{lemma:QuasiFreeGNS}, it is sufficient to consider $h^-$ with $\supp \hti^-\subset\overline{V_-}$. Considering the last equation, we can apply Lemma \ref{lemma:QuasiFreeGNS} again, and see that we may restrict to $h^+$ with $\supp(\hti^+)^*\subset \overline{V_-}$, or, equivalently, $\supp\hti^+\subset \overline{V_+}$.

	So let $h^\pm$ have the specified momentum space supports, and consider the equation in question in the form \eqref{eq:LocalityInOmega}. In view of the support properties of $f,g'$, we can apply Proposition \ref{proposition:RShiftOperator}. Introducing the abbreviations $F^-:=h^-\ot_\rho f$, $F^+:=f\ot_\rho h^+$, $G^-:=h^-\ot_{\rho'}g'$, $G^+:=g'\ot_{\rho'} h^+$, we have
	\begin{align*}
		\supp F^- \subset \Rl^{md}\times (W_0)^{\times n}
		\,,\qquad
		\supp G^+ \subset (W_0')^{\times n'} \times \Rl^{m'd}
		\,,\\
		\supp G^- \subset \Rl^{md}\times (W_0')^{\times n'}
		\,,\qquad
		\supp F^+ \subset (W_0)^{\times n} \times \Rl^{md}
		\,.
	\end{align*}
	Application of Lemma \ref{lemma:LocalityAndTheFlip} now yields $\om(F^-\ot G^+)=\om(\tau(F^-\ot G^+))$ with the flip $\tau$ \eqref{eq:FlipTau}. Note that $\tau$ also acts in momentum space by interchanging the two middle variables.

	To complete the proof, we now show that $\om(\tau (F^-\ot G^+))=\om(G^-\ot F^+)$ by exploiting the compatibility of $\om$ with $\rho,\rho'$ in the form expressed in \eqref{eq:fhat}. We can thus multiply $(G^-\ot F^+)^{\sim}$ with various factors of $R(\pm p\cdot Qq)$ without changing its expectation value in $\om$. Explicitly, $\bk\in\Rl^m$, $\bp\in\Rl^n$, $\bp'\in\Rl^{n'}$, $\bk'\in\Rl^{m'}$,
	\begin{align*}
		(G^-\ot F^+)^{\sim}(\bk,\bp',\bp,\bk')
		&=
		\hti^-(\bk)\gti'(\bp')\fti(\bp)\hti^+(\bk')
		\prod_{l=1}^m\prod_{r=1}^{n'} R(-k_l\cdot Qp_r')
		\prod_{l=1}^n\prod_{r=1}^{m'} R(p_l\cdot Qk_r')
		\,,
		\end{align*}
	and we choose once $k=m$ and two point function $\rho_2(p,q)=R(-p\cdot Qq)$ in \eqref{eq:fhat}, and once $k=m+n+n'$ and $\rho_2(p,q)=R(p\cdot Qq)$. Multiplying $(G^-\ot F^+)^\sim$ by these products results in a function $C$, which takes the form
	\begin{align*}
		\Cti(\bk,\bp',\bp,\bk')
		&:=
		(G^-\ot F^+)^{\sim}(\bk,\bp',\bp,\bk')\cdot
		\prod_{l=1}^m\left\{\prod_{r=1}^{n'}
		R(k_lQp_r')\cdot\prod_{r=1}^{n}R(k_lQp_r)\cdot\prod_{r=1}^{m'}
		R(k_lQk_r')\right\}\times
		\\
		&\qquad\qquad\qquad\;\times
		\prod_{r=1}^{m'}\left\{\prod_{l=1}^m
		R(-k_lQk_r')\cdot\prod_{l=1}^{n'}R(-p'_lQk_r')\cdot\prod_{l=1}^{n}
		R(-p_lQk_r')\right\}
		\\
		&=
		\hti^-(\bk)\fti(\bp)\gti'(\bp')\hti^+(\bk')
		\cdot
		\prod_{l=1}^m\prod_{r=1}^{n}R(k_lQp_r)
		\cdot
		\prod_{l=1}^{n'}\prod_{r=1}^{m'}R(-p_l'Qk_r')
		\\
		&=
		((h^-\ot_{\rho}f)\ot (g'\ot_{\rho'}
		h^+))^{\sim}(\bk,\bp,\bp',\bk')
		\\
		 &=
		(F^-\ot G^+)^{\sim}(\bk,\bp,\bp',\bk')
		\,,
	\end{align*}
	that is, $C=\tau(F^-\ot G^+)$. By construction of $C$, we have $\om(C)=\om(G^-\ot F^+)$. Thus we arrive at
	\begin{align*}
		\om(F^-\ot G^+)
		=
		\om(\tau(F^-\ot G^+))
		=
		\om(C)
		=
		\om(G^-\ot F^+)
		\,,
	\end{align*}
	establishing \eqref{eq:LocalityInOmega}.
\end{proof}

\section{Fock space representations}\label{section:FockSpace}

As shown in the previous section, there exists a large class of multiplicative deformations on $\BU$ which are compatible with all quasi-free Wightman states, and therefore give rise to wedge-local deformations of generalized free field theories. In this section, we will for simplicity consider the explicit two point function
\begin{align*}
	\omti_2(p,q)
	=
	\delta(p+q)\,\eps_\bp^{-1}\,\delta(p^0-\eps_\bp)
	\,,\qquad
	\eps_\bp=\sqrt{\bp^2+m^2},\qquad p=(p^0,\bp)\in\Rl^d\,,
\end{align*}
with some fixed mass $m>0$, and discuss multiplicative deformations in the corresponding GNS representation. We will use the notation from Section \ref{section:generaldeformations}, but generally drop the index $\om$ on $\phi_\om(f),\Psi_\om(f), \DD_\om,\Hil_\om,\Om_\om,U_\om$, since we are working with a fixed state here.

Recall that without deformation, the GNS representation $(\phi,\Hil,\Om)$ of $(\BU,\om)$ describes the model theory of a free scalar field of mass $m$. The representation space $\Hil$ is the Bose Fock space over the single particle space $\Hil_1:=L^2(\Rl^d,d\mu)$ with measure $d\mu(p)=\eps_\bp^{-1}\,\delta(p^0-\eps_\bp)dp$, and the implementing vector $\Om$ is the Fock vacuum.

As a consequence of the Poincar\'e invariance properties of $\om$, there exists an (anti-)unitary representation $U$ of $\PG_+$ on $\Hil$, which leaves $\Om$ invariant, satisfies the spectrum condition, and acts according to $U(x,\La)\Psi(f)=\Psi(\alpha_{x,\La}f)$, $f\in\BU$. Explicitly, we have, $\Psi\in\Hil$,
\begin{align}
	\label{eq:UxLa}
	(U(x,\La)\Psi)_n(p_1,...,p_n)
	&=
	e^{i(p_1+...+p_n)\cdot x}\,\Psi_n(\La^{-1}p_1,...,\La^{-1}p_n)
	\,,\\
	(U(0,j)\Psi)_n(p_1,...,p_n)
	&=
	\overline{\Psi_n(-jp_1,...,-jp_n)}
	\label{eq:J}
	\,.
\end{align}

In the following, we will consider a fixed multiplicative deformation $\rho\in\R_0$, given by a deformation function $R$ and an admissible matrix $Q$. To keep track of both these objects, we will denote the fields representing $(\BU,\ot_\rho)$ as $\phi_{R,Q}(f)$ instead of $\phi^\rho(f)$.

\begin{proposition}\label{proposition:PhiQBasicProperties}
	Let $R$ be a deformation function, $Q$ an admissible matrix, and $f,g\in\BU$, $\Psi\in\DD$.
	  \begin{enumerate}
		 \item\label{item:PhiQDomainAndContinuity} $\phi_{R,Q}(f)$ is a closable operator containing $\DD$ in its domain for any $f\in\BU$, and the map $\BU\ni f\mapsto \phi_{R,Q}(f)\Psi\in\Hil$ is linear and continuous for any $\Psi\in\DD$.
		 \item \label{item:PhiQRepresentation} For $f,g\in\BU$,
			\begin{align}
				\phi_{R,Q}(f)\Psi(g)
				&=
				\Psi(f\ot_{\rho(R,Q)} g)
				\,,\\
				\phi_{R,Q}(f)\phi_{R,Q}(g)
				&=
				\phi_{R,Q}(f \ot_{\rho(R,Q)}g)
				\,,\\
				\phi_{R,Q}(f)^*
				&\supset
				\phi_{R,Q}(f^*)
				\label{eq:PhiQStar}
				\,,\\
				\phi_{R,Q}(f)\Om
				&=
				\phi(f)\Om\,.
				\label{eq:PhiQOnOm}
			\end{align}
	  	\item \label{item:PhiQCovariance} Covariance: For $(x,\La)\in\PG_+^\uparrow$, we have
	  	\begin{align}\label{eq:PhiQCovariance}
	  		U(x,\La)\phi_{R,Q}(f)U(x,\La)^{-1}
	  		=
	  		\phi_{R,\La Q\La^{-1}}(\alpha_{x,\La}f)
	  		\,,
	  	\end{align}
	  	and the reflection at the edge of $W_0$ acts according to
		\begin{align}\label{eq:PhiJCovariance}
			U(0,j)\phi_{R,Q}(f)U(0,j)
			=
			\phi_{R,-Q}(\alpha_j f)
			\,.
		\end{align}
		\item Wedge-Locality: \label{item:PhiQWedgeLocality} Let $f\in\BU(W_0+a), g\in\BU(W_0'+a)$ for some $a\in\Rl^d$. Then
		\begin{align}\label{eq:CommutatorPhiQPhi-Q}
			[\phi_{R,Q}(f),\,\phi_{R,-Q}(g)]\Psi=0\,.
		\end{align}
		 \item \label{item:PhiQReehSchlieder} Reeh-Schlieder property: For any open set $\OO\subset\Rl^d$, the subspace
		  \begin{align}
			\DD_{R,Q}(\OO)
			:=
			\phi_{R,Q}(\BU(\OO))\Om
		\end{align}
		is dense in $\Hil$.
		 \item\label{item:PhiQKleinGordon} $\phi_{R,Q}$ is a weak solution of the Klein-Gordon equation: For $f_1\in C_0^\infty(\Rl^d)$,
		 \begin{align}
		 	\phi_{R,Q}((\Box+m^2)f_1)=0
		 	\,.
		 \end{align}
	  \end{enumerate}
\end{proposition}
\begin{proof}
	The statements in \refitem{item:PhiQRepresentation} follow directly from Proposition \ref{proposition:GNSRepresentationOfCompatibleDeformation} because $\rho$ is compatible with $\om$, and $\phi_{R,Q}$ is a ${}^*$-representation of $(\BU,\ot_\rho)$. \refitem{item:PhiQDomainAndContinuity} Clearly, each $\phi_{R,Q}(f)$ is defined on the dense domain $\DD$, and in view of \eqref{eq:PhiQStar} closable. For $g\in\BU$, the map $f\mapsto\phi_{R,Q}(f)\Psi(g)=\Psi(\rho^{-1}(\rho(f)\ot\rho(g)))$ is linear and continuous because $\rho$, $\rho^{-1}:\BU\to\BU$ and $\Psi:\BU\to\Hil$ are linear and continuous.

	The covariance statements in \refitem{item:PhiQCovariance} follow from $U(x,\La)\Psi(f)=\Psi(\alpha_{x,\La}f)$, $U(0,j)\Psi(f)=\Psi(\alpha_j f)$ and \eqref{eq:PoincareTransformsOfRQProduct}. {\em iv)} In view of the translation covariance {\em iii)}, it is sufficient to show \eqref{eq:CommutatorPhiQPhi-Q} for $a=0$. But this is just a reformulation of Theorem \ref{theorem:Locality}. The Reeh-Schlieder property \refitem{item:PhiQReehSchlieder} is known to hold for the undeformed fields, corresponding to $Q=0$. But in view of \eqref{eq:PhiQOnOm}, $\DD_{R,Q}(\OO)\supset\DD_0(\OO)$, and the density of $\DD_{R,Q}(\OO)$ follows.

	The undeformed field $\phi$ is known to be a weak solution of the Klein-Gordon equation. Using \eqref{eq:PhiQOnOm} again, we therefore have  $\phi_{R,Q}((\Box+m^2)f_1)\Om=\phi((\Box+m^2)f_1)\Om=0$. Now, since $f_1$ has compact support, we find $a\in\Rl^d$ such that $W_0'+a$ lies spacelike to $\supp f_1$. For $g\in\BU(W_0'+a)$, we have in view of \refitem{item:PhiQWedgeLocality}
	\begin{align*}
		\phi_{R,Q}((\Box+m^2)f_1)\phi_{R,-Q}(g)\Om
		=
		\phi_{R,-Q}(g)\phi_{R,Q}((\Box+m^2)f_1)\Om
		=
		0
		\,.
	\end{align*}
	Thus $\phi_{R,Q}((\Box+m^2)f_1)$ vanishes on $\DD_{R,-Q}(W_0'+a)$. As this subspace is dense by \refitem{item:PhiQReehSchlieder}, we arrive at $\phi_{R,Q}((\Box+m^2)f_1)=0$.
\end{proof}

As explained in Section \ref{section:generaldeformations}, we have now constructed a wedge-local quantum field theory, given by the ${}^*$-algebra $\pol_R$ generated by all $\phi_{R,Q}(f)$, $f\in\BU(W_0)$, and
\begin{align}\label{eq:PRQNet}
	\pol_R(\La W_0+x)
	:=
	U(x,\La)\pol_RU(x,\La)^{-1}
	\,,\qquad
	(x,\La)\in \PG_+\,.
\end{align}
In view of the transformation property \eqref{eq:PhiQCovariance}, the algebra $\pol_R(\La W_0+x)$ is generated by the field operators $\phi_{R,\pm\La Q\La^{-1}}(f)$, $f\in\BU(\La W_0+x)$, where the sign "$\pm$" refers to orthochronous / anti-orthochronous Lorentz transformations. In particular, $\pol_R(W_0')$ is generated by all $\phi_{R,-Q}(f)$, $f\in\BU(W_0')$. Thus the orbit $\Q:=\{\La Q\La^{-1}\,:\,\La\in\LG_+\}$ provides a coordinatization for the different directions of the wedges \cite{GrosseLechner:2007, BuchholzLechnerSummers:2010}, whereas the deformation function $R$ labels the kind of deformation used.

It is also possible to proceed from this net of ${}^*$-algebras of unbounded operators to a corresponding net of von Neumann algebras on $\Hil$, generated by bounded functions of the fields. However, as we are interested in a field theoretic setting here, we refrain from giving any details.
\\
\\
Before we proceed to studying the observable consequences of the deformation, we point out that with $Q$, also the rescaled matrices $\la\cdot Q$, $\la\geq0$, are admissible. We have thus constructed one-parameter families $\pol_{R,\la}$ of wedge algebras, representing the deformation maps $\rho(R,\la\cdot Q)$. Taking the limit $\la\to0$ reproduces the undeformed field operators.
\begin{proposition}\label{proposition:PhiDependsContinuouslyOnParameter}
	Let $R$ be a deformation function and $Q$ an admissible matrix. Then, for any $f\in\BU,\Psi\in\DD$,
	\begin{align}
		\lim_{\la\to0}\phi_{R,\la\cdot Q}(f)\Psi
		=
		\phi(f)\Psi
		\,.
	\end{align}
\end{proposition}
\begin{proof}
	Since any $\Psi\in\DD$ is of the form $\Psi=\Psi(g)$, $g\in\BU$, we have $\phi_{R,\la\cdot Q}\Psi=\Psi(f\ot_{\rho(R,\la\cdot Q)}g)$. The claim now follows from the continuity of $\Psi:\BU\to\Hil$ and Proposition \ref{proposition:RQProductContinuousInParameter}.
\end{proof}

We now want to compute the deformed field operators $\phi_{R,Q}(f)$ more explicitly in terms of twisted creation and annihilation operators. To this end, we have to introduce some more notation. For $f_1\in\Ss_1$, we denote by $f_1^\pm(p):=\fti_1(\pm p)$, $p\in H^+_m$, the restriction of the Fourier transform of $f_1$ to the upper and lower mass shell $H^\pm_m$. With this notation, $\Psi(f_1)=f_1^+\in\Hil_1$, and the undeformed field operator has the familiar form
\begin{align}
	\phi(f_1)
	&=
	\ad(\Psi(f_1))+a(\Psi(f_1^*))
	=
	\ad(f_1^+)+a(\overline{f_1^-})
	\,,\qquad
	f_1\in\Ss_1\,.
\end{align}
Here $a,\ad$ form the standard representation of the canonical commutation
relations on $\Hil$. For $\Psi\in\DD$, $\varphi,\psi\in\Hil_1$,
\begin{align}
	(a(\varphi)\Psi)_n(p_1,...,p_n)
	&:=
	\sqrt{n+1}\int d\mu(q)\,\overline{\varphi(q)}\,\Psi_{n+1}(q,p_1,...,p_n)
	\,,
	\\
	\ad(\varphi)
	&:=
	a(\varphi)^*
	\,,
	\\
	[a(\varphi),a(\psi)]&=0,\qquad
	[\ad(\varphi),\ad(\psi)]=0,\qquad
	[a(\varphi),\ad(\psi)]\Psi
	=
	\langle\varphi,\psi\rangle\cdot\Psi
	\,.
\end{align}
We will also work with the distributional kernels $a^\#(p)$ of these operators,
related to $a^\#(\varphi)$ by $\ad(\varphi)=\int d\mu(p)\,\varphi(p)\ad(p)$ and
$a(\varphi)=\int d\mu(p)\,\overline{\varphi(p)}a(p)$, with the commutation
relations,
\begin{align}
	[a(p),a(q)]=0
	\,,\qquad
	[\ad(p),\ad(q)]=0
	\,,\qquad
	[a(p),\ad(q)]
	=
	\eps_\bp\,\delta(p-q)\cdot 1
	\,.
\end{align}

To define deformed versions of these creation/annihilation operators, we introduce the operator-valued function
\begin{align}
	T_R:\Rl^d &\to\B(\Hil),\\
	(T_R(x)\Psi)_n(p_1,...,p_n)
	&:=
	\prod_{k=1}^n R(x\cdot p_k)\,\Psi_n(p_1,...,p_n)
	\,.
\end{align}
It is not difficult to see that quasi-free translationally invariant states are invariant under the shift automorphisms $\tau_x^R$ \eqref{eq:tauR}, and the operators $T_R(x)$ defined above implement these automorphisms on the GNS space. We will however not need these facts here, and only point out that because of the properties \eqref{eq:R-Relations} of $R$, the operator $T_R(x)$ is unitary for any $x\in\Rl^d$, and
\begin{align}
	T_R(x)^*=T_R(-x)=T_R(x)^{-1}
	\,,\qquad
	T_R(0)=1\,.
\end{align}
The operators $T_R(x)$ are now used to twist the canonical commutation relations. We define the operator-valued distributions
\begin{align}\label{eq:ARQ}
	a_{R,Q}(p)
	:=
	a(p)T_R(Qp)
	\,,\qquad
	\ad_{R,Q}(p)
	=
	\ad(p)\,T_R(-Qp)
	\,.
\end{align}
Making use of the antisymmetry of $Q$ and $R(0)=1$, it is straightforward to check that $a(p)$ and $T_R(Qp)$ commute, and thus $\ad_{R,Q}(p)=a_{R,Q}(p)^*$. Explicitly, the deformed annihilation operator acts as, $\varphi\in\Hil_1$, $\Psi\in\DD$,
\begin{align}\label{eq:a_Q}
	(a_{R,Q}(\varphi)\Psi)_n(p_1,...,p_n)
	&=
	\sqrt{n+1}\int d\mu(q)\,\overline{\varphi(q)}\,\prod_{k=1}^n R(Qq\cdot
p_k)\,\Psi_{n+1}(q,p_1,...,p_n)
	\,,
\end{align}
and $\ad_{R,Q}(\varphi)=	a_{R,Q}(\varphi)^*$. It is instructive to compute the exchange relations of the kernels \eqref{eq:ARQ} for different matrices $Q,Q'$. By straightforward calculation, one gets $a(p) T_R(x)=R(x\cdot p)\cdot T_R(x)a(p)$, and hence
\begin{align}\label{eq:RTwistedCCR}
	a_{R,Q}(p)a_{R,Q'}(p')
	&=
	R(p\cdot Qp')R(p\cdot Q'p')\,a_{R,Q'}(p')a_{R,Q}(p)
	\\
	\ad_{R,Q}(p) \ad_{R,Q'}(p')
	&=
	R(p\cdot Qp')R(p\cdot Q'p')\,\ad_{R,Q'}(p')\ad_{R,Q}(p)
	\nonumber
	\\
	a_{R,Q}(p) \ad_{R,Q'}(p')
	&=
	R(-p\cdot Qp')R(-p\cdot Q'p')\,\ad_{R,Q'}(p') a_{R,Q}(p)
	+
	\eps_\bp\,\delta(\bp-\bp')T_R(Qp)T_R(-Q'p)
	\,.
	\nonumber
\end{align}
This exchange algebra generalizes the relations of the Moyal-twisted CCR from \cite{GrosseLechner:2007}. Putting $Q'=Q$, these commutation relations are reminiscent of the Zamolodchikov-Faddeev algebra \cite{ZamolodchikovZamolodchikov:1979,Faddeev:1984}, an observation that will be discussed in Section \ref{section:IntegrableModels}. We also note that for $Q'=-Q$, one can use $R(-a)=R(a)^{-1}$ to simplify the above commutators to
\begin{align*}
	[a_{R,Q}(p),\,a_{R,-Q}(p')]
	&=
	0
	\,,
	\\
	[\ad_{R,Q}(p),\,\ad_{R,-Q}(p')]
	&=
	0\,,
	\\
	[a_{R,Q}(p),\,\ad_{R,-Q}(p')]
	&=
	\eps_\bp\,\delta(\bp-\bp')T_R(Qp)^2
	\,.
\end{align*}
As $-Q$ corresponds to the reflected wedge $W_0'$, these exchange relations and the analytic properties of $R$ can also be used for a proof of the wedge-locality in Proposition \ref{proposition:PhiQBasicProperties} {\em iv)} along the same lines as in \cite{Lechner:2003}.

\begin{proposition}\label{proposition:DeformedSingleFieldOperators}
	The deformed field operators  $\phi_{R,Q}(f_1)$, $f_1\in\Ss_1$, have the form
	\begin{align}\label{eq:PhiQ-AA*}
		\phi_{R,Q}(f_1)
		&=
		\ad_{R,Q}(f_1^+)
		+
		a_{R,Q}(\overline{f_1^-})
		\,.
	\end{align}
\end{proposition}
\begin{proof}
	Let $f_1\in\Ss_1$ with $\supp \fti_1\subset V_-$, and $g\in\Ss_{n+1}(V_+)$,
	$n\in\Nl_0$. Then $\Psi(g)\in\Hil_{n+1}$, $\Psi(f_1\ot_\rho g)\in\Hil_n$, and for $p_1,...,p_n\in H^+_m$, we find, $d\mu(\bp):=d\mu(p_1)\cdots d\mu(p_{n+2})$,
	\begin{align*}
		\phi_{Q,R}(f_1)\Psi(g)
		&=
		\Psi(f_1\ot_{R,Q} g)
		\\
		&=
		\int	d\mu(\bp)\,(f_1\ot_{R,Q}g)^{\sim}(-p_1,p_2,...,p_{n+2})
		a(p_1)\ad(p_2)\cdots\ad(p_{n+2})\Om
		\\
		&=
		\int	d\mu(\bp)\,\fti_1(-p_1)\gti(p_2,..,p_{n+2})
		\prod_{r=2}^{n+2}R(-p_1\cdot Qp_r)
		a(p_1)\ad(p_2)\cdots\ad(p_{n+2})\Om
		\\
		&=
		\int	d\mu(\bp)\,\fti_1(-p_1)\gti(p_2,..,p_{n+2})
		a(p_1)T_R(Qp_1)\ad(p_2)\cdots\ad(p_{n+2})\Om
		\\
		&=
		a_{R,Q}(\overline{f_1^-})\Psi(g)
		\,.
	\end{align*}
	As $g$ and $n$ were arbitrary, we have shown that $\phi_{R,Q}(f_1)$ and $a_{R,Q}(\overline{f_1^-})$ coincide on $\DD$. Since $\supp f_1$ does not intersect the upper mass shell, $f_1^+=0$, and hence the above equation confirms \eqref{eq:PhiQ-AA*}. Taking adjoints, one also finds, $\Psi\in\DD$,
	\begin{align*}
		\phi_{R,Q}(f_1^*)\Psi
		=
		\phi_{R,Q}(f_1)^*\Psi
		=
		a_{R,Q}(\overline{(f_1^*)^-})^*\Psi
		=
		\ad_{R,Q}((f_1^*)^+)\Psi
		\,.
	\end{align*}
	As $\supp \fti_1^*=-\supp \fti_1\subset V_+$, this implies $\phi_{R,Q}(f_1)\Psi=\ad_{R,Q}(f_1^+)\Psi$ for all $\Psi\in\DD$ and all $f_1\in\Ss_1$ with $\supp\fti\subset V_+$. A function $f_1\in\Ss_1$ with arbitrary momentum space support can be decomposed according to $f_1=g_1+h_1$ with the support of $\gti_1$ (respectively $\hti_1$) not intersecting the upper (respectively lower) mass shell. By linearity, this gives \eqref{eq:PhiQ-AA*}.
\end{proof}

It is interesting to note that the deformed field operators can also be expressed as integrals over undeformed fields, similar to the warped convolutions studied in  \cite{BuchholzLechnerSummers:2010}. More precisely, one has, $f_1\in\Ss_1$,
\begin{align}\label{eq:PhiRQAsIntegral}
	\phi_{R,Q}(f_1)
	=
	(2\pi)^{-d}\int dp\,dx\,e^{-ip\cdot x}\,U(x,1)\phi(f_1)U(-x,1)T_R(-Qp)
	\,.
\end{align}
This integral exists as a weak oscillatory integral on vectors $\Psi\in\DD$. In fact, for $\supp f_1\subset V_-$ and $\Psi\in\Hil$, we obtain, $n\in\Nl_0$, $q_1,...,q_n\in H^+_m$,
\begin{align*}
	(2\pi)^{-d}\int dp&\,dx\,e^{-ip\cdot x}\,(U(x,1)\phi(f_1)U(-x,1)T_R(-Qp)\Psi)_n(q_1,...,q_n)
	\\
	&=
	\sqrt{n+1}(2\pi)^{-d}\int dp\,dx\,e^{-ip\cdot x}\int d\mu(q_0)\,\fti_1(-q_0)e^{-iq_0\cdot x}(T_R(-Qp)\Psi)_{n+1}(q_0,q_1,...,q_n)
	\\
	&=
	\sqrt{n+1}\int d\mu(q_0)\,\fti_1(-q_0)(T_R(Qq_0)\Psi)_{n+1}(q_0,q_1,...,q_n)
	\\
	&=
	\sqrt{n+1}\int d\mu(q_0)\,\fti_1(-q_0)\prod_{r=1}^n R(Qq_0\cdot q_r) \Psi_{n+1}(q_0,q_1,...,q_n)
	\\
	&=
	(a_{R,Q}(\overline{f_1^-})\Psi)_n(q_1,...,q_n)
	\,,
\end{align*}
and an analogous calculation can be carried out for the creation operator, establishing \eqref{eq:PhiRQAsIntegral}. However, the integral formula \eqref{eq:PhiRQAsIntegral} reproduces the higher deformed fields $\phi_{R,Q}(f_n)$, $n\geq2$, only if $R$ is of the exponential form $R(a)=e^{ica}$. In this case, $T_R(x)=U(x,1)$, and \eqref{eq:PhiRQAsIntegral} coincides with the warped convolution of $\phi(f)$ by the translation representation $U|_{\Rl^d}$ \cite{BuchholzLechnerSummers:2010}. But for generic $R$, the integrals \eqref{eq:PhiRQAsIntegral} are non-local operators\footnote{I acknowledge helpful discussions with Sergio Yuhjtman about this question.}, and the deformation map $\phi(f_n)\mapsto\phi_{R,Q}(f_n)$ takes a different form. The extension of this map to bounded operators and its integral representations will be discussed in a forthcoming publication with J.~Schlemmer.
\\
\\
We now show that the deformation $\phi(f)\mapsto\phi_{R,Q}(f)$ produces in fact new models, which are not equivalent to their undeformed counterparts. To this end, we will compute the two-particle scattering of the deformed models defined by the fields $\phi_{R,Q}$, following the Haag-Ruelle-Hepp approach \cite{Araki:1999,Hepp:1965}  in its form adapted to wedge-localized operators \cite{BorchersBuchholzSchroer:2001}. Picking $f_1,g_1\in\Ss_1$, the fields $\phi_{R,Q}(f_1)$, $\phi_{R,-Q}(g_1)$ are localized in the wedges $W_0+\supp f_1$ and $W_0'+\supp g_1$, respectively, and create single particle states from the vacuum\footnote{In fact, these fields are temperate polarization-free generators in the sense of \cite{BorchersBuchholzSchroer:2001}.}:
\begin{align}
	\phi_{R,\pm Q}(f_1)\Om=\phi(f_1)\Om=f_1^+\in\Hil_1
	\,.
\end{align}
To define two-particle scattering states, we choose $f_1,g_1$ in such a way that $\supp \fti_1$, $\supp \gti_1$ are concentrated around points on the upper mass shell, and do not intersect the lower mass shell. Furthermore, we introduce the usual notations $\Gamma(f_1):=\{(1,\bp/\eps_\bp)\,:\,p\in\supp\fti_1\}$ for the velocity support of $f_1$, and $f_{1,t}(x):=(2\pi)^{-d/2}\int dp\,\fti(p)e^{i(p^0-\eps_\bp)t}e^{-ip\cdot x}$, with  $p=(p^0,\bp)$ and $\eps_\bp=(\bp^2+m^2)^{1/2}$, for its Klein-Gordon time evolution. It is well known that for asymptotic times $t$, the support of $f_{1,t}$ is essentially contained in $t\Gamma(f_1)$ \cite{Hepp:1965}, that is, the restriction of $f_{1,t}$ to the complement of an open neighborhood of $t\Gamma(f_1)$ converges to zero in the topology of $\Ss_1$ as $|t|\to\infty$.

Because of the compact supports of $f_1,g_1$ in momentum space, the fields $\phi_{R,Q}(f_1)$ and $\phi_{R,-Q}(g_1)$ are not sharply localized in Minkowski space. However, for asymptotic times we have localization of $\phi_{R,Q}(f_{1,t})$  and $\phi_{R,-Q}(g_1)$ in $W_0+t\Gamma(f_1)$ and $W_0'+t\Gamma(g_1)$, respectively. If the velocity supports of $f_1,g_1$ lie in a suitable relative position to the wedge $W_0$, namely $\Gamma(f_1)-\Gamma(g_1)\subset W_0$, these regions are spacelike for $t>0$. As $t\to\infty$, we therefore find two-particle outgoing scattering states as the limits \cite{BorchersBuchholzSchroer:2001}
\begin{align}
	\lim_{t\to\infty}\phi_{R,-Q}(g_{1,t})\phi_{R,Q}(f_{1,t})\Om
	=
	\lim_{t\to\infty}\phi_{R,Q}(f_{1,t})\phi_{R-Q}(g_{1,t})\Om
	=:
	f_1^+\times_{\rm out}^R g_1^+
	\,.
\end{align}
To construct scattering states of incoming particles, the ordering of $f_1,g_1$ has to be reversed: For $t<0$, the localization regions $W_0+t\Gamma(g_1)$ and $W_0'+t\Gamma(f_1)$ lie spacelike if $\Gamma(f_1)-\Gamma(g_1)\subset W_0$, and we have
\begin{align}
	\lim_{t\to-\infty}\phi_{R,-Q}(f_{1,t})\phi_{R,Q}(g_{1,t})\Om
	=
	\lim_{t\to-\infty}\phi_{R,Q}(g_{1,t})\phi_{R,-Q}(f_{1,t})\Om
	=:
	f_1^+\times_{\rm in}^R g_1^+
	\,.
\end{align}
All these limits are easy to compute in the present setting. Since the supports of $f_1,g_1$ do not intersect the lower mass shell, the annihilation parts of the fields drop out, and because the $t$-dependence of $f_{1,t}$ is trivial on the upper mass shell, one finds, $\Gamma(f_1)-\Gamma(g_1)\subset W_0$,
\begin{align*}
	f_1^+\times_{\rm out}^R g_1^+
	&=
	\lim_{t\to\infty}\phi_{R,Q}(f_{1,t})\phi_{R,-Q}(g_{1,t})\Om
	=
	\ad_{R,Q}(f_1^+)\ad(g_1^+)\Om
	\,,\\
	f_1^+\times_{\rm in}^R g_1^+
	&=
	\lim_{t\to-\infty}\phi_{R,-Q}(f_{1,t})\phi_{R,Q}(g_{1,t})\Om
	=
	\ad_{R,-Q}(f_1^+)\ad(g_1^+)\Om
	\,.
\end{align*}
These two-particle vectors have the explicit form
\begin{align*}
	(f_1^+\times_{\rm out/in}^R g_1^+)(p_1,p_2)
	&=
	\left(\ad_{R,\pm Q}(f_1^+)g_1^+\right)_2(p_1,p_2)
	\\
	&=
	\frac{1}{\sqrt{2}}
	\Big(R(\pm p_1\cdot Qp_2) f_1^+(p_1)g_1^+(p_2)
	+
	R(\pm p_2\cdot Qp_1)\,f_1^+(p_2)g_1^+(p_1)
	\Big)
	\,.
\end{align*}
To compute S-matrix elements, let $f_1,g_1,h_1,k_1\in\Ss_1$ with $\Gamma(f_1)-\Gamma(g_1)\subset W_0$, $\Gamma(h_1)-\Gamma(k_1)\subset W_0$. Taking into account these momentum space supports yields the scalar products
\begin{align}\label{eq:SMatrix}
	\langle f_1^+\times_{\rm out}^R g_1^+,\,h_1^+\times_{\rm in}^R k_1^+\rangle
	&=
	\int d\mu(p_1)\,d\mu(p_2)\,R(-p_1\cdot Qp_2)^2\,
	\overline{\fti_1(p_1)}\overline{\gti_1(p_2)}
	\hti_1(p_1)\kti_1(p_2)
	\,.
\end{align}
This formula shows that the S-Matrix elements of the discussed model depends on the deformation. In particular, the scattering in the undeformed theory, corresponding to $R(a)=1$, and the deformed one is different, and the deformed theory is not equivalent to the undeformed one.

Equation \eqref{eq:SMatrix} also clarifies the role of the function $R$ on which our deformation is based: The elastic two-particle S-Matrix kernels of the undeformed and deformed theory differ by its square $R(-p_1\cdot Qp_2)^2$. Since $R$ is a phase factor, the effects in collision processes are relatively small, and can only be measured in special setups such as time delay experiments. These features are similar to the properties of the S-matrices found in the warped convolution deformation \cite{GrosseLechner:2008,BuchholzSummers:2008}.

In view of the dependence of the S-matrix on $Q$, which is only invariant under the boosts preserving $W_0$, but not the full Lorentz group in $d>1+1$ dimensions, we also observe that the two-particle S-matrix obtained here is not fully Lorentz invariant in $d>1+1$. As a consequence, it follows that the model theory constructed here can not contain many observables localized in bounded spacetime regions $\OO$. When passing to von Neumann algebras of observables localized in $\OO$, one finds that at least the Reeh-Schlieder property is violated.

In $d=1+1$ dimensions, however, the identity component of the Lorentz group consists just of the one-dimensional boost group, and hence the above S-matrix is fully Lorentz invariant in this case. We will discuss the two-dimensional situation in Section \ref{section:IntegrableModels}.

\section{Modular structure}\label{section:modular}

In this section we explain how to pass from the unbounded field operators $\phi_{R,Q}(f)$, $f\in\BU$, to associated von Neumann algebras, and study their modular structure. The first step is to control commutators of bounded functions of fields.

\begin{proposition}\label{proposition:PhiIsEssentiallySelfadjoint}
	Let $R$ be a deformation function, and $Q$ an admissible matrix.
	\begin{propositionlist}
		\item Let $f_1=f_1^*\in\Ss_1$. Then $\phi_{R,Q}(f)$ is essentially self-adjoint.
		\item Let $f_1=f_1^*\in\Ss_1(W_0)$ and $g_1=g_1^*\in\Ss_1(W_0')$. Then the self-adjoint closures $\overline{\phi_{R,Q}(f_1)}$ and $\overline{\phi_{R,-Q}(g_1)}$ commute, i.e.,
		\begin{align}\label{eq:CommutatorOfExponentiatedFields}
			\left[e^{it\overline{\phi_{R,Q}(f_1)}},\;e^{is\overline{\phi_{R,-Q}(g_1)}}\right]=0\,,\qquad t,s\in\Rl\,.
		\end{align}
	\end{propositionlist}
\end{proposition}
\begin{proof}
	{\em i)} We will first show that any $\Psi\in\DD$ is an entire analytic vector for the field operators $\phi_{R,Q}(f_1)$, $f_1\in\Ss_1$. For $\varphi\in\Hil_1$, the annihilation operator $a_{R,Q}(\varphi)$ can be estimated with the help of \eqref{eq:a_Q} and $|R(t)|=1$ as
	\begin{align*}
		\left|(a_{R,Q}(\varphi)\Psi_n)(p_1,...,p_{n-1})\right|
		&\leq
		\sqrt{n}\left|\int d\mu(q)\,\varphi(q)\,\prod_{k=1}^n R(Qq\cdot p_k)\,\Psi_n(q,p_1,...,p_{n-1})\right|
		\\
		&\leq
		\sqrt{n}\int d\mu(q)\,|\varphi(q)|\,|\Psi_n(q,p_1,...,p_{n-1})|
		\,.
	\end{align*}
	By standard $L^2$-estimates, this implies $\|a_{R,Q}(\varphi)|_{\Hil_n}\|\leq\sqrt{n}\|\varphi\|$, and taking adjoints, also $\|\ad_{R,Q}(\varphi)|_{\Hil_n}\|\leq\sqrt{n+1}\|\varphi\|$ follows. Thus we have the basic bound
	\begin{align}\label{eq:PhiNBound}
		\|\phi_{R,Q}(f_1)|_{\Hil_n}\|
		\leq
		\sqrt{n+1}\,(\|f_1^+\|_{\Hil_1}+\|f_1^-\|_{\Hil_1})\,.
	\end{align}
	With this bound one can easily show that any $\Psi\in\DD$ is an entire analytic vector for $\phi_{R,Q}(f_1)$ (see, for example, the proof of Theorem X.41 in \cite{ReedSimon:1975}). As $\DD\subset\Hil$ is dense, application of Nelson's analytic vector theorem \cite[Thm.~X.39]{ReedSimon:1975} shows that $\phi_{R,Q}(f_1)$, $f_1^*=f_1$, is essentially self-adjoint. Its self-adjoint closure will be denoted $\overline{\phi_{R,Q}(f_1)}$.

	{\em ii)} Using the bound \eqref{eq:PhiNBound} again, one also shows that $e^{is\overline{\phi_{R,-Q}(g_1)}}\Psi$, $g_1=g_1^*\in\Ss_1$, $s\in\Rl$, $\Psi\in\DD$, is an entire analytic vector for $\overline{\phi_{R,Q}(f_1)}$, as in free field theory. Hence on $\Psi\in\DD$, the commutator \eqref{eq:CommutatorOfExponentiatedFields} can be computed as the power series
	\begin{align*}
		\left[ e^{it\overline{\phi_{R,Q}(f_1)}},\,e^{is\overline{\phi_{R,-Q}(g_1)}}\right]\Psi
		=
		\sum_{n,n'=0}^\infty \frac{i^{n+n'}t^ns^{n'}}{n!n'!}\,
		\left[\phi_{R,Q}(f_1)^n,\,\phi_{R,-Q}(g_1)^{n'}\right]\Psi\,.
	\end{align*}
	As $\phi_{R,Q}(f_1)$ and $\phi_{R,-Q}(g_1)$ commute on $\phi(\BU)\Om$ (Proposition \ref{proposition:PhiQBasicProperties} \refitem{item:PhiQWedgeLocality}), the proof is finished.
\end{proof}

We now introduce the von Neumann algebras generated by the self-adjoint field operators,
\begin{align*}
	\M_{R,Q} &:=  \left\{e^{i\overline{\phi_{R,Q}(f_1)}}\,:\,f_1=f_1^*\in\Ss_1(W_0)\right\}''\,,
	\\
	\Mhat_{R,Q} &:= \left\{e^{i\overline{\phi_{R,-Q}(g_1)}}\,:\,g_1=g_1^*\in\Ss_1(W_0')\right\}''\,.
\end{align*}
In view of Proposition \ref{proposition:PhiIsEssentiallySelfadjoint} {\em ii)}, these algebras commute, $\Mhat_{R,Q}\subset{\M_{R,Q}}'$. By standard arguments making use of the Reeh-Schlieder property established in Proposition \ref{proposition:PhiQBasicProperties} \refitem{item:PhiQReehSchlieder}, it also follows that the vacuum vector is cyclic for $\M_{R,Q}$ and $\Mhat_{R,Q}$. As these algebras commute, $\Om$ is separating as well. Thus Tomita Takesaki modular theory applies to the pair $(\M_{R,Q},\Om)$, and provides us with modular unitaries $\Delta_{R,Q}^{it}$ and a modular involution $J_{R,Q}$. In the following theorem, we show that these data are stable under the deformation, i.e. do not depend on $R$ and $Q$ within the specified limitations. For the special case $R(a)=e^{ia}$, this fact was already shown in \cite{BuchholzLechnerSummers:2010}.

\begin{theorem}\label{theorem:StabilityOfModularData}
	Let $R$ be a deformation function and $Q$ an admissible matrix.
	\begin{propositionlist}
		\item The modular data $J_{R,Q}, \Delta_{R,Q}$ of $\M_{R,Q},\Om$ are independent of $R$ and $Q$.
		\item The Bisognano-Wichmann property holds,
		\begin{align}\label{eq:BiWi}
			\Delta_{R,Q}^{it} = U(0,\La_1(2\pi t))\,,\qquad J_{R,Q}=U(0,j)\,,
		\end{align}
		with $\La_1(t):(x^0,...,x^{d-1})\mapsto(\cosh(t)x^0+\sinh(t)x^1,\sinh(t)x^0+\cosh(t)x^1,x^2,...,x^{d-1})$ denoting the boosts in $x^1$-direction.
		\item $\Mhat_{R,Q}={\M_{R,Q}}'$.
	\end{propositionlist}
\end{theorem}
\begin{proof}
	We first show that given $f\in\BU(W_0)$, the closed operator $F:=\overline{\phi_{R,Q}(f)}$ is affiliated with $\M_{R,Q}$. To this end, let $\Psi\in\dom F$, $\Psi_0\in\DD$, and consider a real test function $g_1'\in\Ss_1(W_0')$. As $F^*$ changes the particle number only be a finite amount, both $\Psi_0$ and $F^*\Psi_0$ are entire analytic vectors for $G':=\overline{\phi_{R,-Q}(g_1')}$. Taking also into account that $F^*$ and $(G')^p$ commute on $\DD$ for any $p\in\Nl_0$ (Proposition \ref{proposition:PhiQBasicProperties} \refitem{item:PhiQWedgeLocality}), we find
	\begin{align*}
		\langle\Psi_0,\,e^{iG'}F\Psi\rangle
		&=
		\langle e^{-iG'}\Psi_0,\,F\Psi\rangle
		=
		\sum_{p=0}^\infty \frac{(-i)^p}{p!}\langle F^*(G')^p\Psi_0,\,\Psi\rangle
		=
		\sum_{p=0}^\infty \frac{(-i)^p}{p!}\langle (G')^pF^*\Psi_0,\,\Psi\rangle
		\\
		&=
		\langle e^{-iG'}F^*\Psi_0,\,\Psi\rangle
		=
		\langle \Psi_0,\,Fe^{iG'}\Psi\rangle
		\,.
	\end{align*}
	As $\DD\subset\Hil$ is dense, this implies $e^{iG'}F\Psi=Fe^{iG'}\Psi$. Clearly, this identity then also holds when $e^{iG'}$ is replaced by any operator in the ${}^*$-algebra $\A$ generated (algebraically) by the $e^{i\overline{\phi_{R,-Q}(g_1')}}$, $g_1'\in\Ss_1(W_0')$ real. But any $A'\in\M_{R,Q}'$ is a weak limit of a sequence $A_n'$ in $\A$, and $A_n'F\Psi=FA_n'\Psi$ is stable under weak limits. Thus we arrive at $A'F\Psi=FA'\Psi$ for all $\Psi\in\dom F$, i.e., $F$ is affiliated with $\M_{R,Q}$.

	Proceeding to the polar decomposition $F=V|F|$ and the spectral projections $E_n$ of $|F|$ onto spectrum in the interval $[0,n]$, we have $V,E_n|F|\in\M_{R,Q}$ for all $n\in\Nl$. Now let $S_{R,Q}$ denote the Tomita operator of $(\M_{R,Q},\Om)$. As $S_{R,Q}VE_n|F|\Om=|F|E_nV^*\Om$, the strong convergence $E_n\to1$ as $n\to\infty$ and the closedness of $S_{R,Q}$ imply that $F\Om$ lies in the domain of $S_{R,Q}$, and $S_{R,Q}F\Om=F^*\Om$.

	As all these considerations apply in particular to the special case $R=1$, we have now gathered sufficient information for establishing {\em i)}. Let $S$ denote the Tomita operator of the undeformed algebra $\M:=\M_{1,Q}$ w.r.t. $\Om$, and let $f\in\BU(W_0)$ as above. Making use of  \eqref{eq:PhiQOnOm}, we find
	\begin{align*}
		S_{R,Q}\phi(f)\Om
		=
		S_{R,Q}\phi_{R,Q}(f)\Om
		=
		\phi_{R,Q}(f)^*\Om
		=
		\phi_{R,Q}(f^*)\Om
		=
		\phi(f^*)\Om
		=
		\phi(f)^*\Om
		=
		S\phi(f)\Om
		\,,
	\end{align*}
	i.e., $S_{R,Q}$ and $S$ coincide on the subspace $\phi(\BU(W_0))\Om$. But this domain is a core for $S=J\Delta^{1/2}$ because it is dense and the modular group $\Delta^{it}$ acts as the Lorentz boosts $\La_1(2\pi t)$ which leave $W_0$ invariant \cite{BisognanoWichmann:1975}. As $S$ and $S_{R,Q}$ are closed operators, this shows that $S_{R,Q}$ is an extension of $S$, i.e., $S_{R,Q} \supset S$.

	We now consider the commutants $\M_{R,Q}'$, $\M'$. By modular theory, their Tomita operators w.r.t. $\Om$ are the adjoints $S_{R,Q}^*$, $S^*$. In complete analogy to above, one can show that for $f'\in\BU(W_0')$, the operator $\overline{\phi_{R,-Q}(f')}$ is affiliated with $\M_{R,Q}'$, and
	\begin{align*}
		S_{R,Q}^*\phi(f')\Om
		&=
		S_{R,Q}^*\phi_{R,-Q}(f')\Om
		=
		\phi_{R,-Q}(f')^*\Om
		=
		\phi(f')^*\Om
		=
		S^*\phi(f')\Om
		.
	\end{align*}
	Since $\phi(W_0')\Om$ is a core for $S^*$, this shows $S_{R,Q}^*\supset S^*$, or, equivalently, $S_{R,Q}=S_{R,Q}^{**}\subset S^{**}=S$. Together with the previously established extension $S_{R,Q}\supset S$, this yields $S_{R,Q}=S$. The identities $\Delta_{R,Q}=\Delta$ and $J_{R,Q}=J$ then follow from the uniqueness of the polar decomposition $S_{R,Q}=J_{R,Q}\Delta_{R,Q}^{1/2}$.

	As the Bisognano-Wichmann property \eqref{eq:BiWi} is known to hold for the free field theory \cite{BisognanoWichmann:1975}, {\em ii)} follows immediately from {\em i)}. The transformation law \eqref{eq:PhiJCovariance} of the field implies $U(0,j)\M_{R,Q}U(0,j)=\Mhat_{R,Q}$ by extension from analytic vectors. By Tomita's theorem, this yields
	\begin{align*}
		\Mhat_{R,Q}=U(0,j)\M_{R,Q}U(0,j)
		=
		J_{R,Q}\M_{R,Q}J_{R,Q}
		=
		{\M_{R,Q}}'\,,
	\end{align*}
	which proves {\em iii)}.
\end{proof}

Von Neumann algebras with modular data identical to the geometric ones found in free field theory have been studied as a possible tool in the construction of quantum field theories before \cite{Wollenberg:1992, LeitzMartiniWollenberg:2000}. It is therefore interesting to note that the deformation construction presented here establishes a new infinite family of solutions to this inverse problem in modular theory.
\\
\\
For a formulation of our models in the framework of algebraic quantum field theory, we now consider the von Neumann algebras
\begin{align}\label{eq:ARW}
	\A_R(W) := U(x,\La)\M_{R,Q}U(x,\La)^{-1}\,,
\end{align}
where $W$ is a wedge and $(x,\La)\in\PG_+$ is any Poincar\'e transformation satisfying $\La W_0+x=W$. (We have suppressed the dependence of the left hand side on $Q$ here because $Q$ is transformed by $\La$.) The transformation behaviour of the field $\phi_{R,Q}(f)$ implies that \eqref{eq:ARW} is well-defined, i.e. independent of the choice of $(x,\La)$. Furthermore, we have, $W,\Wti\in\W$,
\begin{align*}
	\A_R(W)&\subset\A_R(\Wti)\quad\text{ for } W\subset \Wti\,,\\
	\A_R(W)&=\A_R(W')'\,,\\
	U(x,\La)\A_R(W)U(x,\La)^{-1}
	&=
	\A_R(\La W+x)\,,
\end{align*}
where in the last line, $(x,\La)\in\PG_+$ is arbitrary. In view of the unitarity of $U$, it is also clear that $\Om$ is cyclic and separating for each $\A_R(W)$, $W\in\W$. We summarize these findings in the following proposition.

\begin{proposition}
	Let $R$ be a deformation function and $Q$ an admissible matrix. Then the map $\A_R:\W\ni W\longmapsto\A_R(W)\subset\B(\Hil)$ \eqref{eq:ARW} is an isotonous, Haag-dual net of von Neumann algebras which transforms covariantly under the adjoint action of $U$, and the vacuum vector $\Om$ is cyclic and separating for each $\A_R(W)$, $W\in\W$.
\end{proposition}

As is well known, it is possible to extend a net with the above properties to arbitrary regions in Minkowski space by taking suitable intersections of the algebras $\A_R(W)$. This extension always preserves isotony, locality, and covariance. However, the algebras associated with bounded regions might be small or even trivial.

In the case at hand, the results found in the computation of the two-particle scattering states in the previous section imply that $\Om$ is not cyclic for algebras associated with bounded regions in dimension $d>1+1$. In two space-time dimensions, the situation is however different. This case will be discussed in the next section.

\section{Integrable models as deformations of free field theories}\label{section:IntegrableModels}

Up to this point, the dimension $d\geq1+1$ of spacetime did not play any role in our constructions. Now we will consider the special case $d=1+1$ of a two-dimensional Minkowski space. The matrix $Q$ appearing in the deformation two point function then has the form \eqref{eq:CovariantQ}
\begin{align}\label{eq:2x2Q}
	Q
	&=
	\la\,
	\left(
		\begin{array}{cc}
			0&1\\
			1&0\\
		\end{array}
	\right)
	\,,\qquad \la\in\Rl\,.
\end{align}
In two dimensions, it is convenient to parametrize the upper mass shell of mass $m>0$ by the rapidity $\te\in\Rl$ according to $p(\te):=m(\cosh\te,\,\sinh\te)$.
Inserting this parametrization into the deformation two point function \eqref{eq:Rho2FromRAndQ} yields
\begin{align}
	\rho_2(p(\te_1),p(\te_2))
	&=
	R(-p(\te_1)\cdot Qp(\te_2))
	=
	R(\la m^2\sinh(\te_1-\te_2))
	\,,\qquad
	\te_1,\te_2\in\Rl\,,
\end{align}
and we denote the square of this function by
\begin{align}\label{eq:S}
	S_\la:\Rl\to\Cl\,,\qquad
	S_\la(\te)&:= R(\la m^2\sinh\te)^2\,.
\end{align}
As mentioned earlier, $R$ is analytic on the upper half plane because of the half-sided support of its Fourier transform. As the hyperbolic sine is an entire function mapping the strip $S(0,\pi):=\{\zeta\in\Cl\,:\,0<\im\,\zeta<\pi\}$ onto the upper half plane, this implies that $S_\la$, $\la\geq0$, extends to an analytic function on $S(0,\pi)$, with distributional boundary values at $\Rl$ and $\Rl+i\pi$. From the properties \eqref{eq:R-Relations} of $R$ and $\sinh$, it is obvious that
\begin{align}\label{eq:S-Relations}
	\overline{S_\la(\te)}
	=
	S_\la(\te)^{-1}
	=
	S_\la(-\te)
	=
	S_\la(\te+i\pi)
	\,,\qquad
	\la,\te\in\Rl\,.
\end{align}
These relations are well known from the analysis of completely integrable quantum field theories with factorizing S-matrices on two-dimensional Minkowski space \cite{AbdallaAbdallaRothe:1991}, where they express the unitarity, hermitian analyticity, and crossing symmetry \cite{BabujianFoersterKarowski:2006} of a two-particle S-matrix of such a model. Here these properties show up as a consequence of our deformation construction.

Not only the typical relations of a factorizing S-matrix appear here, but also the characteristic algebraic structure known as the Zamolodchikov-Faddeev algebra \cite{ZamolodchikovZamolodchikov:1979,Faddeev:1984}: For the rapidity space creation/annihilation operators $z_\la(\te):=a_{R,Q}(p(\te))$, $\zd_\la(\te):=\ad_{R,Q}(p(\te))$, the relations \eqref{eq:RTwistedCCR} (with both $Q$ and $Q'$ replaced by \eqref{eq:2x2Q}) read
\begin{align*}
	z_\la(\te_1)z_\la(\te_2)
	&=
	S_\la(\te_1-\te_2)\,z_\la(\te_2)z_\la(\te_1)
	\\
	\zd_\la(\te_1) \zd_\la(\te_2)
	&=
	S_\la(\te_1-\te_2)\,\zd_\la(\te_2)\zd_\la(\te_1)
	\\
	z_\la(\te_1)\zd_\la(\te_2)
	&=
	S_\la(\te_2-\te_1)\zd_\la(\te_2) z_\la(\te_1)
	+
	\delta(\te_1-\te_2)\cdot 1
	\,.
\end{align*}
This is precisely the Zamolodchikov-Faddeev algebra. In  the context of factorizing S-matrices, it is mostly used as an auxiliary structure to organize $n$-particle scattering states (see, for example, \cite{CastroAlvaredo:2001}). However, it is also possible to take it as a starting point for the construction of model theories.

This latter point of view has been taken by Schroer, who suggested to use the fields $\phi_\la(x):=\int d\te\,(e^{ip(\te)\cdot x}\zd_\la(\te)+e^{-ip(\te)\cdot x}z_\la(\te))$ as wedge-local polarization-free generators for constructing quantum field theories \cite{Schroer:1997-1}. Although this construction was originally formulated independently of deformation ideas, the same fields also appear in the present setting, and coincide with the deformed fields $\phi_{R,Q}$ from the previous section. In the two-dimensional context, their properties as listed in Proposition \ref{proposition:PhiQBasicProperties} were known already in case the scattering function $S$ satisfies \eqref{eq:S-Relations} and is analytic and bounded on the strip $S(0,\pi)$ \cite{Lechner:2003}.

Full-fledged quantum field theories based on these deformed fields have been constructed in the framework of algebraic quantum field theory \cite{Haag:1996}: After passing from the wedge-local fields to corresponding nets of von Neumann algebras, operator-algebraic techniques become available for the analysis of the local observable content of these models \cite{BuchholzLechner:2004}. We recall from \cite{Lechner:2006,Lechner:2008} that if $S$ is {\em regular} in the sense that it has a bounded analytic extension to the strip $\{\zeta\in\Cl\,:\,-\eps<\im\,\zeta<\pi+\eps\}$ for some $\eps>0$, then the quantum field theory generated by $\phi_\la$ contains observables localized in double cones, at least for the radius of the double cone above some minimal size. In fact, there exist so many such local observables that they generate dense subspaces from the vacuum, as it is typical in quantum field theory (Reeh-Schlieder property). Also all other standard properties of quantum field theory are satisfied by these models, and the factorizing S-matrix with scattering function $S$ can be recovered from their $n$-particle collision states \cite{Lechner:2008}. We note this relation between multiplicative deformations and integrable models as the following theorem.

\begin{theorem}
	On two-dimensional Minkowski space, every integrable quantum field theory with scattering function $S$ of the form \eqref{eq:S} arises from a free field theory by a (multiplicative) deformation. If $S$ is regular, then the deformed theory is local in the sense that the vacuum is cyclic for all observable algebras associated with double cones above a minimal size \cite[Thm. 5.6]{Lechner:2008}.
\end{theorem}

Although the structure of integrable quantum field theories is quite simple, the important message for the deformation technique presented here is that this method is capable of deforming covariant local {\em free} quantum field theories to covariant local {\em interacting} quantum field theories. For the deformed models to contain sufficiently many local observables, we only have to select the deformation function $R$ in such a way that $S$ \eqref{eq:S} is regular. For example, this is the case for the finite Blaschke products
\begin{align}
	R(a)
	=
	\prod_{k=1}^N
	\frac{z_k-a}{z_k+a}
	\,,
\end{align}
where the zeros $z_1,...,z_N$ lie in the upper half plane and occur in pairs $z_k,-\overline{z_k}$ \eqref{eq:RExampleFunctions}.

\section{Conclusions}\label{section:Conclusions}

In this paper we have established a family of deformations of quantum field theories, leading to new models with non-trivial interaction in any number of space-time dimensions $d$. This result supports the general deformation approach, and shows that it is possible to use deformation methods for obtaining interacting local field theories from models without interaction. As interacting quantum field theories in physical spacetime must necessarily involve particle production processes \cite{Aks:1965}, and particle production was ruled out here because of the relatively simple form of the multiplicative deformations, the obtained models are not yet physically realistic. In two space-time dimensions, they have the structure of integrable models, and there are indications that the family of integrable models which can be realized in this manner is actually much larger\footnote{S.~Alazzawi, C.~Schützenhofer, work in progress.}. For models on higher-dimensional Minkowski space, however, one needs to allow for particle production processes already on the level of the deformation maps, and replace the multiplicative deformations by more general integral operators \eqref{eq:RhoNMKernels}. Apart from these modifications, it seems to be possible to use the same approach as presented here to realize also interactions with momentum transfer and particle production by deformation methods.

From a structural point of view, it is desirable to uncouple the deformations from the specific form of the Borchers-Uhlmann tensor algebra. This has been achieved in the case of the warped convolutions \cite{BuchholzSummers:2008}, which are formulated in such a way that they are applicable to any vacuum quantum field theory \cite{BuchholzLechnerSummers:2010}. Such an operator-algebraic reformulation of the deformations studied here is currently under investigation.

Regarding the operator-algebraic structure, we have shown that the modular data of the von Neumann algebras associated with wedge-local deformed quantum fields represented in compatible states are identical to those in the undeformed theory. This is essentially a consequence of the compatibility of the deformations with the ${}^*$-involution and unit element of $\BU$, and can therefore be expected to be a generic feature of deformations of quantum field theories. This feature connects our deformation approach to another approach to the construction of quantum field theories, based on the inverse problem in modular theory \cite{LeitzMartiniWollenberg:2000}. Furthermore, also the root of the S-matrix plays a role in both, our present setting, where it appears in the deformation two-point function, and in the context of inverse problems in modular theory, where it is used to identify modular conjugations  \cite{Wollenberg:1992}. These interesting connections require further investigation, which will be presented elsewhere.

\subsection*{Acknowledgements}

I enjoyed helpful discussions in the Vienna deformation group, in particular with Sabina Alazzawi, Jan Schlemmer, Jakob Yngvason, and Sergio Yuhjtman. Many thanks go also to Stefan Waldmann for informing me about the rigidity of tensor algebras.



\end{document}